\documentclass[10pt,twocolumn,letterpaper]{article}

\usepackage{cvpr}
\usepackage{times}
\usepackage{epsfig}
\usepackage{graphicx}
\usepackage{amsmath}
\usepackage{amssymb}
\usepackage{color}
\usepackage{booktabs,threeparttable}
\usepackage{colortbl}
\usepackage[export]{adjustbox}
\usepackage{caption}
\usepackage{multirow}
\usepackage{amssymb}
\usepackage{amsthm}
\usepackage{subfigure} 
\usepackage{algorithmic}
\usepackage{algorithm}
\usepackage{epstopdf}
\newtheorem{theorem}{Theorem}

\newtheorem{lemma}{Lemma}

\theoremstyle{definition}

\theoremstyle{remark}
\newtheorem{remark}{Remark}

\usepackage[pagebackref=true,breaklinks=true,letterpaper=true,colorlinks,bookmarks=false]{hyperref}

\cvprfinalcopy 

\def\cvprPaperID{} 

\begin{document}
\newenvironment{sequation}{\begin{equation}\small}{\end{equation}}
\title{Distributed Fusion Estimation for Stochastic Uncertain Systems with Network-Induced Complexity and Multiple Noise}

\author{Li Liu$^*$\\
School of Information Science\\
and Electrical Engineering,\\
Ludong University, Yantai 264025, China\\
{\tt\small liulildu@163.com}
\and
Wenju Zhou\\
School of Mechatronic Engineering\\
and Automation, Shanghai University,
\\Shanghai 200072, China\\
{\tt\small zhouwenju2004@126.com}
\and
Minrui Fei\\
School of Mechatronic Engineering\\
and Automation, Shanghai University,
\\Shanghai 200072, China\\
{\tt\small mrfei@staff.shu.edu.cn}
\and
Zhile Yang\\
Shenzhen Institute of Advanced Technology,\\
Chinese Academy of Sciences,\\
Shenzhen 518055, Guangdong, China\\
{\tt\small zyang07@qub.ac.uk}
\and
Hongyong Yang\\
School of Information Science\\
and Electrical Engineering,\\
Ludong University, Yantai 264025, China\\
{\tt\small hyyang@ldu.edu.cn}
\and
Huiyu Zhou\\
School of Informatics,\\
University of Leicester, United\\
{\tt\small hz143@leicester.ac.uk}
}

\maketitle

\begin{abstract}
This paper investigates an issue of distributed fusion estimation under network-induced complexity and stochastic parameter uncertainties. First, a novel signal selection method based on event-trigger is developed to handle network-induced packet dropouts as well as packet disorders resulting from random transmission delays, where the ${H_2}/{H_\infty }$ performance of the system is analyzed in different noise environments. In addition, a linear delay compensation strategy is further employed for solving the complexity network-induced problem, which may deteriorate the system performance. Moreover, a weighted fusion scheme is used to integrate multiple resources through an error cross-covariance matrix. Several case studies validate the proposed algorithm and demonstrate satisfactory system performance in target tracking.
\end{abstract}

\newenvironment{shrinkeq}[1]
{ \bgroup
  \addtolength\abovedisplayshortskip{#1}
  \addtolength\abovedisplayskip{#1}
  \addtolength\belowdisplayshortskip{#1}
  \addtolength\belowdisplayskip{#1}}
{\egroup\ignorespacesafterend}
\section{Introduction}

Energy Internet has been a new and advanced paradigm established by effective integration of energy and information network infrastructures involving traditional centralized generation \cite{Bib1}, distributed energy resources \cite{Bib2}, advanced communication \cite{Bib3}, smart metering \cite{Bib4}, intelligent computing \cite{Bib5}, and smart management systems \cite{Bib6}. The structure of Energy Internet has witnessed promising solutions with improved efficiency to holistic energy flow problems, while accommodating the transformation of traditional power networks. However, owing to the requirements of growing system scale, in order to address the challenges in system scalability and reliability, a new distributed networked control system (DNCS) needs to be developed.

In the existing DNCS structure, the data or signals are exchanged between three system components (e.g. sensors, controllers and actuators) using a shared communication network. Furthermore, the system component is physically distributed and interconnected with the others in order to coordinate their tasks for achieving the desired objectives \cite{Bib7}. Thus, distributed controllers are capable of coordinating their behaviors by transmitting/receiving information to/from other controllers within a certain neighboring area. The communication networks are introduced to the distributed systems \cite{Bib8,Bib9} and this inevitably produces network-induced complexity (such as random transmission delays, packet dropouts, packet disorders, and missing/fading measurements), which may significantly deteriorates the system performance \cite{Bib10}. Therefore, it is a challenging issue for Energy Internet to incorporate distributed information with network-induced complexity in controller design and time-sensitive applications \cite{Bib11}.

In engineering practice, noise can be categorized into three types: bounded and stochastic uncertainty noise, energy-bounded noise and uncertain white noise \cite{Bib34}. In ship control applications \cite{Bib12_3}, for example, we may expect different environmental changes (e.g. winds, waves and currents) during the ship navigation, leading to bounded and stochastic disturbance uncertainties for the control units. At the same time, sensor noise may be of unknown characteristics but bounded power, while the instrument input can be mixed with uncertain Gaussian white noise. To effectively establish a working system, stochastic system uncertainty is investigated using a multiplicative noise model \cite{Bib11,Bib12}. The correlated noise (i.e. auto and cross-correlated noise) \cite{Bib12,Bib14} is considered when they are handling the estimation problem over the networks. Moreover, the bounded noise \cite{Bib15,Bib16} is introduced in a realistic networked system to evaluate the measurement error and external disturbances. On the other hand, filtering, in particular ${H_\infty }$ filtering, plays an important role in the field of signal processing and communications. However, considering signal over sensor networks with transmission delays in the communication process, ${H_\infty }$ filtering \cite{Bib8,Bib17} is insensitive to the uncertainty. To solve the parametric uncertainty problem, robust filtering based on finite-horizon Kalman filter \cite{Bib19} is proposed to reduce the conservativeness and maintain the performance of the filter.

For the transmission over networks, the one-step prediction compensation is used for most of existing systems due to the bandwidth constraint to handle transmission delays and/or packet dropouts. Such that the augmented state scheme is applied to obtain optimal estimation \cite{Bib20}. A model for multi-sensor fusion system was established by using the full-rank decomposition approach and to describe the random observation delays and packet dropouts \cite{Bib21}. Moreover, to break the resource constraints (i.e. bandwidth or energy), a measurement re-organization method \cite{Bib22} has been developed to simplify the system model with delays. Another classic method is the modelling of delayed system by sequences of Bernoulli random variable distributions \cite{Bib20, Bib23}, where the measurement model for each sensor is augmented by different characteristics. Due to random delays for the involved sensors, the local estimator based on the compensation or augmentation method possesses high computation cost.

To reduce the unnecessary resource consumption, the delayed system is established by introducing the event-triggered scheme. The containment control for multi-agent system was investigated under the centralized event-triggered containment algorithm to deal with constant time delays \cite{Bib25}. Li \emph{et al.} \cite{Bib26} presented an event-triggered sampled-data stabilization for switching linear systems to obtain the average dwell time conditions. Moreover, packet disorders inevitably appear due to transmission delays. Several works on dealing with packet disorders have been developed. To store the data packets, the signal selection method of zero-order-holder (ZOH) chooses and receives the most recent arriving data packets \cite{Bib27}, whilst the signal sequence re-ordering method using the logic ZOH scheme \cite{Bib8,Bib28,Bib29} was introduced to choose the latest time-stamped data packets. It is worth pointing out that the logic ZOH based on event-trigger mechanism is widely applied in the field of networked control systems (NCSs). However, the issue of the fixed-time event-triggered consensus was studied for the multi-agent systems \cite{Bib28_2}. Therefore, handling signals with or without packet disorders, as well as estimating filter parameters are difficult problems. In addition, with the aid of the event-triggered scheme, designing the distributed fusion estimation is a complex problem.

For the sake of alleviating the negative influence during transmission, some promising estimation and filtering methods have been proposed for distributed systems. Liu \emph{et al.} \cite{Bib30} introduced a two level weighted fusion scheme using the reorganized innovation sequence as well as filtering error cross-covariance, and presented a distributed weighted fusion estimation method with transmission delay and cross-correlated noise. Different distributed fusion schemes, and event-triggered Kalman consensus filter \cite{Bib32} was used to deal with the distributed estimation. Considering reducing energy consumptions, the finite-horizon filtering \cite{Bib22,Bib33} derives the filter parameters from the upper bound for the covariance matrix of the estimation error. Furthermore, for a class of discrete stochastic systems, a finite horizon state estimation is investigated, which using the event-based modelling strategy to actively dropped the disordered packet \cite{Bib35}. To guarantee optimal estimation performance, Chen \emph{et al.} \cite{Bib34} studied the distributed fusion estimation to handle admissible parameter uncertainties. The information fusion estimation method based on bio-inspired computation was presented for the designing dynamic 3D positioning system \cite{Bib37_4}. Since the distributed estimation is converted into a convex optimization problem, the estimated state has better robustness for the transmission performance.

For a class of stochastic uncertain networked systems, this paper intends to design the local filter with packet dropouts and disorders resulting from random transmission delays, as well as multiple noise resources. Due primarily to the difficulty in handling the disordered and missing data packets, dealing with stochastic and deterministic parameter uncertainties simultaneously. Motivated by the above analysis, the purpose of this paper concentrates on establishing a new system description, analyzing system performance and designing local/fusion estimator. The main contributions of this paper are summarized as follows:

(i) Suppose that the proposed system model contains energy-bounded noise and uncertain white noise. Due to the difficulty of estimating the appropriate filter parameters, the logic ZOH scheme is widely used in NCSs \cite{Bib8,Bib28,Bib29,Bib36,Bib34_2}. Based on the signal selection method for the logic ZOH, the event-triggered condition is introduced, and a system model is established to effectively deal with the network-induced complex and improve system performance. The proposed system model synthetically describes the network-induced phenomena and simplifies the characterization of the complex network environments.

(ii) For the energy-bounded noise and uncertain white noise, studying the optimization problem is difficult as the objective function is of significant nonlinear items \cite{Bib16,Bib34,Bib42}. The assumption on the parameter uncertainty can be satisfied by analyzing the mixed ${H_2}/{H_\infty }$ estimation performance. Subsequently, different from using sequences of Bernoulli random variables distribution \cite{Bib20,Bib36_3,Bib37_3,Bib24}, a simplified local estimation scheme in the augmented state-space with the re-organized measurements is presented, aiming to consume fewer network resources and provide the appropriate filter parameters.

(iii) To handle packet dropouts and missing data, instead of adopting the augmented state approach \cite{Bib20,Bib40} of compensating every step delay, a linear delay compensation for packet dropouts is firstly presented to suppress the growing error accumulation and alleviate the computational burden by re-ordering measurement sequences. Moreover, to further improve the estimation accuracy, the missing data packets due to the packet dropouts are collected using the one-step prediction method.

The rest of this paper is organized as follows: Establish the stochastic uncertain model using the logic ZOH scheme in Section 2. Section 3 shows a local estimator and analyses the estimation performance. Section 4 elaborates the distributed fusion estimation approach. A numerical example is presented in Section 5 and Section 6 draws the conclusions.

\textbf{Notations}: Throughout this paper ${{\mathbb{R}}^{r}}$ and ${{\mathbb{R}}^{r\times r}}$ denote the $r$ and $r\times r$ -dimensional Euclidean space, respectively. $E\left(  \cdot  \right)$ represents the mathematical expectation operator, and the superscript $T$ is the transpose. $M$ is a real symmetric matrix satisfying $M>0$, while ${{M}^{-1}}$ indicates that the inverse of $M$. Moreover, $tr\left( M \right)$ is the trace of $M$, and the symmetric terms in a symmetric matrix are denoted by $*$. $I$ represents the identity matrix with an appropriate dimensions. ${l_2}\left[ {0,\;\infty } \right)$ represents the space of a square integral function on $\left[ {0,\;\infty } \right)$, and ${{\delta }_{k,l}}$ denotes the Kronecker function (i.e. ${{\delta }_{k,l}}=1$ if $k=l$, and ${{\delta }_{k,l}}=0$ if $k\ne l$).

\section{Problem formulation and analysis}
\subsection{System description}
A considered object is measured by $L$ sensors through the stochastic uncertain system (Fig.1 shows that the local estimation is solved by the spatially distributed sensors, subsequently, it is sent to the fusion center).
\begin{figure}[t]
	\begin{center}
		\includegraphics[width=0.5\linewidth]{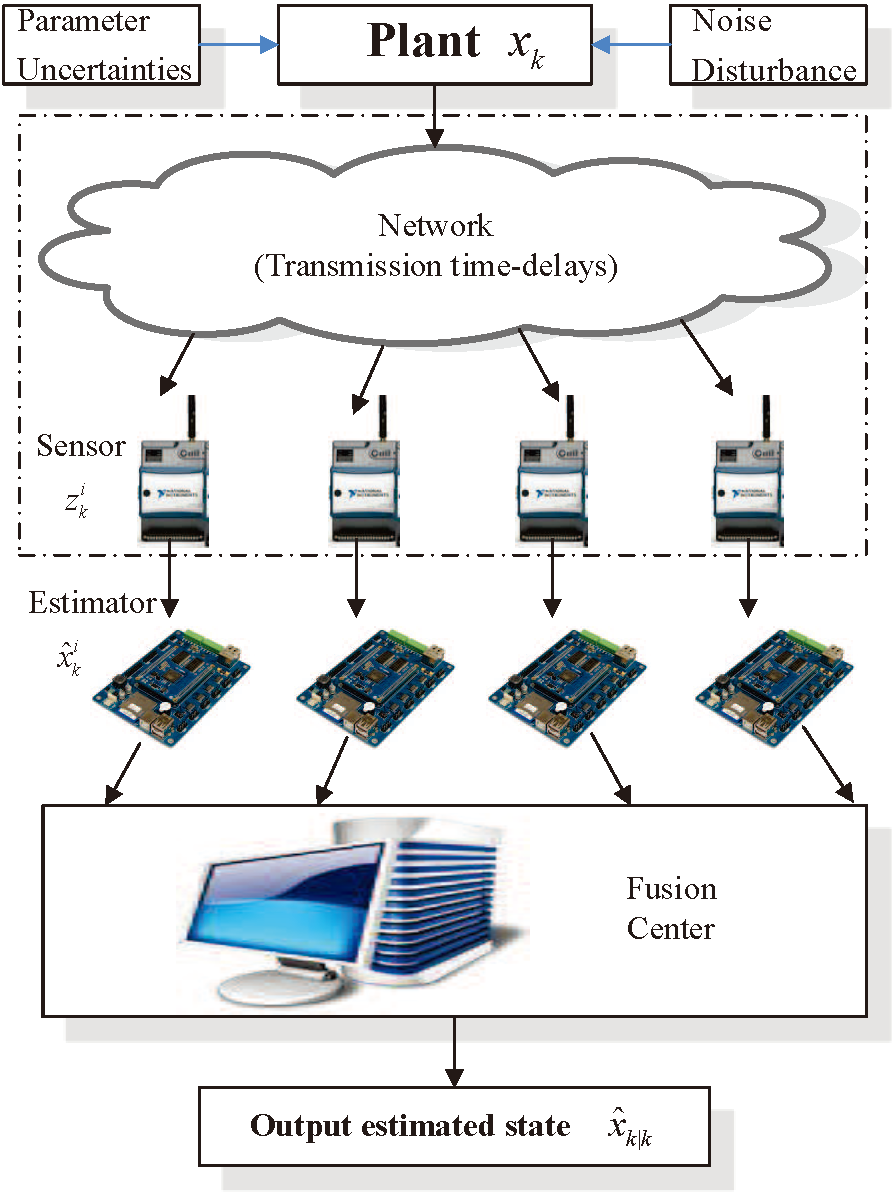}
	\end{center}
    \vspace{-4mm}
	\caption{\small Structure of distributed fusion estimation}
	\label{Fig 1}
\end{figure}

This process is described by the following linear discrete-time model:
\begin{shrinkeq}{-0.8ex}
\begin{sequation}\label{eq:1}
\begin{gathered}
  {x_{k + 1}} = \left( {{A_k} + {\mathcal{F}_k}{F_k}{E_k} + \sum\limits_{\vartheta  = 1}^\hbar  {{A_\vartheta }{\varpi _{\vartheta ,k}}} } \right){x_k} + {B_k}{w_k} \hfill \\
  \;\; + {G_k}{v_k},\;\;\;\;\;\;\;\;\;\;\;\;\;\;\;\;\;\;\;\;\;\;\;\;\;\;\;\;\;\;\;\;\;\;k = 1,2, \cdots  \hfill \\
\end{gathered}
\end{sequation}
\small
\begin{equation}\label{eq:2}
z_k^i = \left( {C_k^i + \mathcal{H}_k^i{F_k}E_k^i} \right){x_k} + B_k^i{w_k} + G_k^i{v_k}\;,\;i = 1,\; \cdots \;,\;L
\end{equation}
\end{shrinkeq}
where ${x_k} \in {\mathbb{R}^r}$ represents the system state, and $z_k^i \in {\mathbb{R}^{{m_i}}}$ refers to the measurement output from the $i^{th}$ sensor at the time instant $k$. It is assumed that the initial state ${x_0}$ with mean ${\mu _0}$ as well as covariance ${P_0}$, which is uncorrelated with other noise. ${A_k} \in {\mathbb{R}^{r \times r}}$, ${B_k} \in {\mathbb{R}^r}$, ${G_k} \in {\mathbb{R}^r}$, $C_k^i \in {\mathbb{R}^{{m_i} \times r}}$, ${\mathcal{F}_k}$, ${E_k}$, $\mathcal{H}_k^i$, $E_k^i$, $B_k^i$ and $G_k^i$ are known time-varying matrices. Note that ${F_k}$ represents the parameter uncertainty satisfying ${F_k}F_k^T \leqslant I$. ${w_k} \in {l_2}\left[ {0,\infty } \right)$ denotes the energy-bounded noise. The multiplied term ($\sum\limits_{\vartheta  = 1}^\hbar  {{A_\vartheta }{\varpi _{\vartheta ,k}}} {x_k}$) given in Eq.(\ref{eq:1}) presents the stochastic parameter uncertainty. ${\varpi _{\vartheta ,k}} \in \mathbb{R}\;\;\left( {\vartheta  = 1,\; \cdots \;,\;\hbar } \right)$ and ${v_k}$ with covariance matrices ${\theta _{\vartheta ,k}}$ and ${R_k}$ denote the Gaussian white noise. Note that they are mutually uncorrelated and independent. In practice, it is difficult to obtain accurate covariance ${\theta _{\vartheta ,k}}$. Therefore, the definition of the lower bound $\theta _\vartheta ^L$ and the upper bound $\theta _\vartheta ^U$ should be carefully determined, which satisfy the following constraint:

\begin{shrinkeq}{-1ex}
\begin{equation}\label{eq:3}
\small
\theta _\vartheta ^L \leqslant {\theta _{\vartheta ,k}} \leqslant \theta _\vartheta ^U\;,
\end{equation}
\end{shrinkeq}
where the lower and the upper bounds are known.

In this paper, for the given bounds $\theta _\vartheta ^L$ and $\theta _\vartheta ^U\;\left( {\vartheta {\text{ = }}1,\; \cdots \;,\;\hbar } \right)$, the following three conditions are simultaneously satisfied \cite{Bib16,Bib34}:

\textbf{Condition 1.} When ${w_k} \equiv 0$ and ${v_k} \equiv 0$, the error estimation system is asymptotically mean-square stable.

\textbf{Condition 2.} Satisfying the zero-initial condition and ${v_k} \equiv 0$, each local estimation error (e.g. $e_k^i$) with an arbitrary energy-bounded noise ${w_k}$ conforms to the following inequality:
\begin{shrinkeq}{-1.5ex}
\begin{equation}\label{eq:4}
\small
\sum\limits_{k = 0}^\infty  {E\left( {{{\left( {e_k^i} \right)}^T}e_k^i} \right)}  < \gamma _i^2\sum\limits_{k = 0}^\infty  {E\left( {w_k^T{w_k}} \right)} ,
\end{equation}
where ${\gamma _i}$ denotes the ${H_\infty }$ disturbance attenuation level bound (DALB) for the estimation error. Note that $e_k^i\;\left( {i = 1,\; \cdots \;,\;L} \right)$ represent the difference between the actual and estimated states, which is defined in Eqs.(\ref{eq:17}) and (\ref{eq:18}).
\end{shrinkeq}

\textbf{Condition 3.} Each local estimation error $e_k^i$ with ${v_k}$ and ${w_k} \equiv 0$ (i.e. the Gaussian white noise and energy-bounded noise) guarantees that the upper bound of the ${H_2}$ performance
\begin{shrinkeq}{-1.5ex}
\begin{equation}\label{eq:5}
\small
{J^i} = {\lim _{k \to \infty }}E\left( {{{\left( {e_k^i} \right)}^T}e_k^i} \right)
\end{equation}
is minimal. Note that the relation between noise ${w_k}$ and ${v_k}$ are mutually uncorrelated and independent.
\end{shrinkeq}

\begin{remark}\label{Remark 1}
The aforementioned performance requirements for the local estimation of each subsystem are analyzed using the unified formulation (Conditions 1-3). Thus the filter parameters for the local estimation will be separately derived. Comparing with the distributed fusion estimation reported in \cite{Bib9,Bib11,Bib32}, the proposed approach considers multiple noise (e.g. white noise and  energy-bounded noise) simultaneously. In this sense, the noise assumption and modelling in Eqs.(\ref{eq:1}) and (\ref{eq:2}) indicates that the distributed fusion estimation approach can be implemented in practice, and will be made more generalised in the future development.
\end{remark}

\subsection{System modelling for signal selection scheme}
\begin{figure*}[t]
	\begin{center}
		\includegraphics[width=0.6\linewidth]{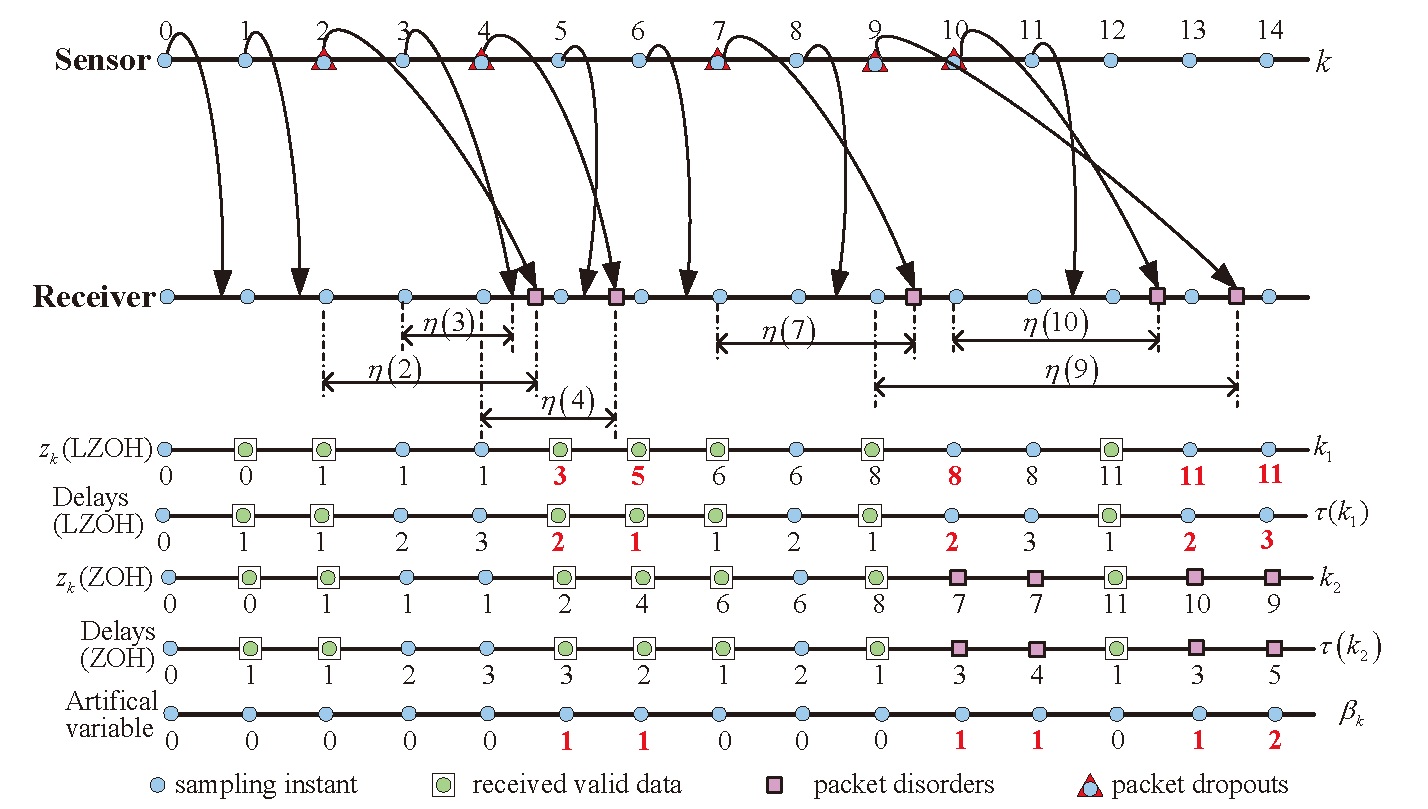}
	\end{center}
    \vspace{-4mm}
	\caption{Sorting procedure of the packet sequence for logic ZOH and ZOH}
	\label{Fig 2}
\end{figure*}

Due to the resource constraints, the network congestion unavoidably interrupts the system evaluation. During the time-stamped data packets transmission, the signal selection scheme is introduced based on the logic ZOH, and the latest data packet is chosen before being transmitted. Therefore, for the processor, the stored signals on the receiver will be updated until it receives the latest data packet, while other data packets are discarded. Furthermore, the packet disorders during the transmission are dropped via the logic ZOH \cite{Bib28}. A typical scenario describing the network-induced complexity is shown in Fig.2, and the detailed analysis is shown in Appendix A. Meanwhile, due to the limited service quality over the network, data packets are transmitted from the sensor to the estimator through a communication network in a single packet manner. Therefore, it is possible that an early leaving packet may arrive at the processor late, which leads to packet disorders. Since the disordered packets usually contain incorrect temporal signals, they should be discarded \cite{Bib36}.

Fig.2 illustrates the relationship between the states and variables. From the previous discussion, to formulate the relationship of current time $k$ and time-stamp, we define time-stamps ${k_1}$ and ${k_2}$ as those when we have the received data packets and the most recent transmission signal, which satisfies:
\begin{shrinkeq}{-1ex}
\begin{sequation}\label{eq:6}
k = \tau \left( {{k_1}} \right) + {k_1} = \tau \left( {{k_2}} \right) + {k_2}\;.
\end{sequation}
\end{shrinkeq}
\begin{remark}\label{Remark 2}
Suppose that the transmission delays are represented as $\eta \left( {{k_1}} \right)$ and $\eta \left( {{k_2}} \right)$ from the sensor to the processor with and without logic ZOH respectively, while $0 \leqslant \eta \left( {{k_1}} \right) \leqslant \eta \left( {{k_2}} \right) \leqslant N$ ($N$ denotes the maximum of the transmission delay) is satisfied. Subsequently, the transmission delays for the received data packets are denoted as $\tau \left( {{k_1}} \right) \in \mathbb{N}$ and $\tau \left( {{k_2}} \right) \in \mathbb{N}$ at the sampling time instant $k$. For the given Eq.(\ref{eq:6}), the sampling time instant for the received data packets consists of two components: the time-stamp and the transmission delay.
\end{remark}

Note that the received latest data packets approximate the current signal, and variable ${\beta _k}$ is used for expressing the relationship between ${k_1}$ and ${k_2}$, that is:
\begin{shrinkeq}{-1ex}
\begin{equation}\label{eq:7}
\small
{k_1} = {k_2} + {\beta _k}\;.
\end{equation}
In which ${\beta _k}$ is determined according to the time-stamp before being transmitted. Since the time-stamp ${k_2}$ for receiving the most recent data packet is not proceeding the latest one ${k_1}$, ${\beta _k} \geqslant 0$ is satisfied. In terms of the transmission delays obtained from Eqs.(\ref{eq:6}) and (\ref{eq:7}), $\tau \left( {{k_1}} \right) = \tau \left( {{k_2}} \right) - {\beta _k}$ is rewritten, and $\tau \left( {{k_1}} \right) \leqslant \tau \left( {{k_2}} \right)$ is satisfied.
\end{shrinkeq}

According to the definition of symbol ${{k}_{2}}$ and the ZOH scheme, it should be pointed out that if no new signal input the ZOH, its output keeps as the constant. For instance, at a specific time instant $k'$ (e.g. $k' = \left\{ {9,\;12,\;13} \right\}$ in Fig.2) as well as time-stamp $t\left( {{k}_{2}} \right)$, if
\begin{shrinkeq}{-1ex}
\begin{equation}\label{eq:8}
\small
t\left( {{k_2} + 1} \right) = k' + 1 - \tau \left( {{k_2} + 1} \right) < k' - \tau \left( {{k_2}} \right) = t\left( {{k_2}} \right)\;,
\end{equation}
the received signal at time instant $k' + 1$ is based on the previous signals with time-stamp $t\left( {{k}_{2}} \right)$, and more updated signals with time-stamp $t\left( {{k}_{2}}+1 \right)$ are not used even if they are available by the ZOH signal selection scheme. Therefore, the most recent transmission signal is kept until the latest packet is used no matter whether it is a packet disorder or not \cite{Bib28}.
\end{shrinkeq}

Then, using the logic ZOH in order to detect the packet disorders, the random transmission delay satisfies:
\begin{shrinkeq}{-1ex}
\begin{equation}\label{eq:9}
\small
\tau \left( {{k}_{1}}+1 \right)\le \tau \left( {{k}_{1}} \right)+1
\end{equation}
and
\begin{equation}\label{eq:10}
\small
t\left( {{k_1} + 1} \right) = k' + 1 - \tau \left( {{k_1} + 1} \right) \ge k' - \tau \left( {{k_1}} \right) = t\left( {{k_1}} \right)\;.
\end{equation}
\end{shrinkeq}
This implies that the transmission signal at time $k' + 1$ is based on the outdated data packet at time-stamp $t\left( {{k}_{1}} \right)$, and a more updated data packet at the time-stamp $t\left( {{k}_{1}}+1 \right)$ is used if it is available at the processor. The estimation accuracy of system is further improved, and the computational burden is also reduced due to the data packets with less transmission delays \cite{Bib36}.

As the logic ZOH is event-driven, a function $f\left( \centerdot ,\ \centerdot  \right)$ for the event generator is defined to determine the event-triggered condition as follows:
\begin{shrinkeq}{-1ex}
\begin{equation}\label{eq:11}
\small
t\left( {{k}_{1}}+1 \right)=t\left( {{k}_{1}} \right)+{{\min }_{j>0}}\left\{ j|f\left( {{\sigma }_{k}},\delta  \right)>0,\ j\in \mathbb{N} \right\}\ ,
\end{equation}
in which
\begin{equation}\label{eq:12}
\small
f\left( {{\sigma _k},\delta } \right) = \sigma _k^T\Omega {\sigma _k} - \delta \hat x_{{k_1}}^T\Omega {\hat x_{{k_1}}}
\end{equation}
and
\begin{equation}\label{eq:13}
\small
{\sigma _k} \buildrel \Delta \over = {\hat x_{{k_1} + j}} - {\hat x_{{k_1}}}\;.
\end{equation}
\end{shrinkeq}
Note that ${\hat x_{{k_1} + j}}$ represents the estimation at the latest event time ${k_1} + j$, and ${\hat x_{{k_1}}}$ is the current filter with the time-stamp ${k_1}$. Furthermore, $\Omega $ is a weighted matrix, which is symmetric positive-definite, and $\delta \in \left[ 0,\ 1 \right)$ is a scalar threshold.

\begin{remark}\label{Remark 3}
In \cite{Bib35_2}, the event-triggered condition is defined in Eq.(\ref{eq:11}), and the range of the threshold parameter $\delta $ is set to be $\left[ {0,\;1} \right)$ in this paper, which satisfies the following relationship:
\begin{shrinkeq}{-1.5ex}
\begin{equation}\label{eq:14}
\small
f\left( {{\sigma _k},\delta } \right) = \;\left[ {\hat x_{{k_1} + j}^T\;\;\;\;\hat x_{{k_1}}^T} \right]\mathbb{Q}\left[ \begin{gathered}
  {{\hat x}_{{k_1} + j}} \hfill \\
  {{\hat x}_{{k_1}}} \hfill \\
\end{gathered}  \right]
\end{equation}
\end{shrinkeq}
where $\small\mathbb{Q} \triangleq \left[ \begin{gathered}
  \Omega \;\;\;\; - \Omega  \hfill \\
   - \Omega \;\;\;\left( {1 - \delta } \right)\Omega \; \hfill \\
\end{gathered}  \right]$. Note that $\delta \hat x_{{k_1}}^T\Omega {\hat x_{{k_1}}}$ is a positive scalar owing to $\Omega  > 0$. On the one hand, if $\delta  \geqslant 1$, $\mathbb{Q}$ is non-positive definite. That is $f\left( {{\sigma _k},\delta } \right) \leqslant 0$ for any ${\hat x_{{k_1} + j}}$ and ${\hat x_{{k_1}}}$, the event-triggered condition given in Eq.(\ref{eq:12}) is never triggered. On the other hand, if $\delta  < 0$, $\mathbb{Q}$ is positive definite, namely, the event-triggered condition in Eq.(\ref{eq:12}) is always triggered. Consequently, the scalar threshold parameter satisfies $\delta  \in \left[ {0,\;1} \right)$.
\end{remark}

Next, the measurement re-organization approach is investigated to simplify the system description as well as alleviate the computational burden. For the $i^{th}$ subsystem using the logic ZOH, when the processor receives the data packet $z_{{t}}^i$ with time-stamp $t$ and transmission delay ${\tau ^i}\left( {{k_1}} \right)$, the stored signal $y_k^i$ is re-organized as
\begin{shrinkeq}{-1ex}
\begin{equation}\label{eq:15}
\small
y_k^i = z_t^i\;,
\end{equation}
\end{shrinkeq}
where $t = k - {\tau ^i}\left( {{k_1}} \right)$ in Eq.(\ref{eq:6}). Note that the processor (i.e. estimator) receiving the time-stamped data packets may accompany the information with packet delays and dropouts at each sampling time \cite{Bib22}. It suggests that using the logic ZOH based on event-trigger, the signals with packet dropouts and packet disorders are translated into random transmission delays using the time-stamped data packets.

\begin{remark}\label{Remark 4}
For the delay systems, the well-known innovation re-organization approach was proposed in \cite{Bib36_2,Bib37_2}, which designed $l+1$ different standard Kalman filtering with the $l$-step time delayed measure. To further reduce the computation cost, the proposed measurement re-organization approach given in Eq.(\ref{eq:15}) is used to handle the random transmission delays, which is translated into a unified form considering the current time instant as well as the time-stamp before being transmitted. Therefore, we design the optimal state estimation for each subsystem to develop a unified finite horizon filtering. More importantly, the network-induced events such as packet dropouts and packet disorders resulting from the random transmission delays are translated into the re-organized measurement sequence, and a signal selection method based on event-trigger is applied to further simplifying the temporal process.
\end{remark}

Taking into account the stored measurement sequence with network-induced complexity, the issue of distributed fusion estimation is transferred into probing the optimal state estimation ${\hat x_{k|k}}$, which is compensated and fused by local estimation $\hat x_{t|t}^i$ with transmission delay $t = k - {\tau ^i}\left( {{k_1}} \right)$.

\subsection{Local estimation design}
In this section, considering the system noise in Eq.(\ref{eq:1}) and the measurement noise in Eq.(\ref{eq:2}) are different, the augmentation method for the state-space model is involved \cite{Bib34}. Meanwhile, the local estimator $\hat x_{k|k}^i\left( {i = 1,\; \cdots \;,\;L} \right)$ is derived by using the Schur complement lemma to analyze the estimation performance \cite{Bib16}.

\subsection{Augmentation method for state vector}
Firstly, the received data packet $y_k^i$ is obtained from Eqs.(\ref{eq:2}) and (\ref{eq:15}), that is
\begin{shrinkeq}{-1ex}
\begin{equation}\label{eq:16}
\small
y_k^i = \left( {C_t^i + \mathcal{H}_t^i{F_t}E_t^i} \right){x_t} + B_t^i{w_t} + G_t^i{v_t}\;,\;i = 1,\; \cdots \;,\;L\;.
\end{equation}
\end{shrinkeq}
It is assumed that the processor has a sufficient capability to find the optimal estimated state $\hat x_{t|t}^i$ from the stored valid signals $\left\{ {y_0^i\;,\; \cdots \;,\;y_{k - 1}^i\;,\;y_k^i} \right\}$.

The objective of the distributed fusion estimation is achieving the optimal estimation $\hat x_{k|k}^i$ from the upper bound of the local error covariance matrix. Thus, the local state is estimated employing the re-organized data packets:
\begin{shrinkeq}{-1ex}
\begin{equation}\label{eq:17}
\small
\hat x_{t|t}^i = \hat x_{t|t - 1}^i + K_t^i\left( {z_t^i - \hat C_t^i\hat x_{t|t - 1}^i} \right)\;,
\end{equation}
\begin{equation}\label{eq:18}
\small
\hat x_{t + 1|t}^i = \hat A_t^i\hat x_{t|t - 1}^i + L_t^i\left( {z_t^i - \hat C_t^i\hat x_{t|t - 1}^i} \right)\;,
\end{equation}
\end{shrinkeq}
where $\hat x_{t|t}^i$ and $\hat x_{t + 1|t}^i$ denote the filter and predictor, respectively. Then the local estimation errors are defined by $\tilde e_t^i \triangleq {x_t} - \hat x_{t|t - 1}^i$ and $e_t^i \triangleq {x_t} - \hat x_{t|t}^i$. Time-varying matrices $\hat C_t^i$, $K_t^i$, $\hat A_t^i$ and $L_t^i$ are the solved filter parameters.

In order to derive the upper bounds for the covariance matrices of the estimation errors, the augmentation method for state vectors including error and estimated state are defined as:
\begin{shrinkeq}{-1.5ex}
\begin{equation}\label{eq:19}
\small
\tilde \Psi _t^i \triangleq \left[ \begin{gathered}
  \tilde e_t^i \hfill \\
  \hat x_{t|t - 1}^i \hfill \\
\end{gathered}  \right]\;,\;\Psi _t^i \triangleq \left[ \begin{gathered}
  e_t^i \hfill \\
  \hat x_{t|t}^i \hfill \\
\end{gathered}  \right]\;.
\end{equation}
\end{shrinkeq}

Subsequently, the augmented system model for the state-space shown in Eqs.(\ref{eq:1}) and (\ref{eq:17})-(\ref{eq:19}) is further discussed in Appendix B. Note that the deterministic uncertainty ${F_t}$, energy-bounded noise ${w_t}$ as well as Gaussian white noise ${v_t}$ and ${\varpi _{\vartheta ,t}}$ appear in Eqs.(B.4) and (B.5). It is difficult to solve the accurate estimation using the covariance matrices $\tilde \Theta _t^i$ and $\tilde \Sigma _{t + 1}^i$. Therefore, this paper investigates an alternative method to solve the upper bounds for the errors. Furthermore, the filter parameters are optimized by minimizing the error covariance.

\subsection{Performance analysis}
To analyse the system performance, the estimation error from the local estimation $\hat x_{k|k}^i$ can be rewritten as:
\begin{shrinkeq}{-1.5ex}
\begin{equation}\label{eq:20}
\small
\begin{gathered}
  \chi _{k + 1}^i = \left( {A_{t3}^i + \Delta A_{t3}^i} \right)\chi _k^i + \left( {\sum\limits_{\vartheta  = 1}^\hbar  {\left( {{\varpi _{\vartheta ,k}}{A_{\vartheta ,t3}}} \right)} } \right)\chi _k^i \hfill \\
  \;\; + B_{t3}^i{w_k} + G_{t3}^i{v_k}\;, \hfill \\
  \tilde e_{k}^i{\text{ = }}D_{t3}^i\chi _k^i\;, \hfill \\
\end{gathered}
\end{equation}
\end{shrinkeq}
where $\chi _{k}^i \triangleq \left[ \begin{gathered}
  {x_{k}} \hfill \\
  \tilde e_{k}^i \hfill \\
\end{gathered}  \right]$, and the parameters are defined as:
\begin{shrinkeq}{-1.5ex}
\[A_{t3}^i = \left[ \begin{gathered}
  {A_k}\;\;\;\;\;\;\;\;\;\;\;\;\;\;\;\;\;\;\;\;\;\;\;\;\;\;\;\;\;\;\;\;\;\;\;\;\;\;\;\;\;\;\;\;\;0 \hfill \\
  {A_k} - \hat A_k^i + L_k^i\left( {\hat C_k^i - C_k^i} \right)\;\;\;\;\hat A_k^i - L_k^i\hat C_k^i \hfill \\
\end{gathered}  \right],\]
\end{shrinkeq}

\begin{shrinkeq}{-1.5ex}
\small
\begin{equation}\label{eq:21}
\begin{gathered}
  \Delta A_{t3}^i = \left[ \begin{gathered}
  {\mathcal{F}_k}{F_k}{E_k}\;\;\;\;\;\;\;\;\;\;\;\;\;\;\;\;\;\;\;\;\;\;0 \hfill \\
  \left( {{\mathcal{F}_k} - L_k^i\mathcal{H}_k^i} \right){F_k}{E_k}\;\;\;\;0 \hfill \\
\end{gathered}  \right],\;{A_{\vartheta ,t3}} = \left[ \begin{gathered}
  {A_\vartheta }\;\;0 \hfill \\
  {A_\vartheta }\;\;0 \hfill \\
\end{gathered}  \right], \hfill \\
  B_{t3}^i = \left[ \begin{gathered}
  {B_k} \hfill \\
  {B_k} - L_k^iB_k^i \hfill \\
\end{gathered}  \right],\;G_{t3}^i = \left[ \begin{gathered}
  {G_k} \hfill \\
  {G_k} - L_k^iG_k^i \hfill \\
\end{gathered}  \right],\;D_{t3}^i = \left[ {0\;\;\;I} \right] \hfill \\
\end{gathered}
\end{equation}
\end{shrinkeq}

First of all, the following bounded real lemmas for the discrete-time stochastic system $\chi _k^i$ must be satisfied.

\begin{lemma}\label{Lemma 1}
\cite{Bib34} For the stochastic uncertain system, the discrete-time model is established as follows:
\begin{sequation}\label{eq:22}
\begin{gathered}
  x{'_{k + 1}} = \left( {A' + \Delta A'v{'_k}} \right)x{'_k} + B'w{'_k}\;, \hfill \\
  z{'_k} = C'x{'_k}, \hfill \\
\end{gathered}
\end{sequation}

\noindent
in which $v{'_k}$ is the Gaussian white noise with covariance matrices $R{'_k}$, and $w{'_k} \in {l_2}\left[ {0,\infty } \right)$ denotes an energy-bounded noise. For a given scalar $\gamma '$ and the ${H_\infty }$ performance, the following inequality
\begin{sequation}\label{eq:23}
\sum\limits_{k = 0}^\infty  {E\left( {{{\left( {z{'_k}} \right)}^T}z{'_k}} \right)}  < {\left( {\gamma '} \right)^2}\sum\limits_{k = 0}^\infty  {E\left( {{{\left( {w{'_k}} \right)}^T}w{'_k}} \right)}
\end{sequation}
\noindent
is held. With the aid of the system performance in Eqs.(\ref{eq:22}) and (\ref{eq:23}), if and only if there exists a matrix ${\rm X}' > 0$, such that the inequality
\begin{sequation}\label{eq:24}
\begin{gathered}
  {\left( {A'} \right)^T}{\rm X}'A' + {\left( {C'} \right)^T}C' + R'{}_k{\left( {\Delta A'} \right)^T}{\rm X}'\left( {\Delta A'} \right) - {\rm X}' \hfill \\
  \;\; + {\left( {A'} \right)^T}{\rm X}'B'{\left( {{{\left( {\gamma '} \right)}^2}I - {{\left( {B'} \right)}^T}{\rm X}'B'} \right)^{ - 1}}{\left( {B'} \right)^T}{\rm X}'A' < 0. \hfill \\
\end{gathered}
\end{sequation}
is satisfied.
\end{lemma}

\begin{lemma}\label{Lemma 2}
\cite{Bib34} Let $\Gamma _1^T = {\Gamma _1}$, ${\Gamma _2}$ and ${\Gamma _3}$ be real matrices with appropriate dimensions, and ${\lambda _k}$ satisfying $\lambda _k^T{\lambda _k} \leqslant I$. Then
\begin{equation}\label{eq:25}
{\Gamma _1} + {\Gamma _3}{\lambda _k}{\Gamma _2} + \Gamma _2^T\lambda _k^T\Gamma _3^T < 0
\end{equation}
is held, if and only if there exists a positive scalar $\sigma  > 0$, such that the following inequality is satisfied:
\begin{equation}\label{eq:26}
\left[ \begin{gathered}
   - \sigma I\;\;\;\sigma {\Gamma _2}\;\;\;0 \hfill \\
  *\;\;\;\;\;\;\;\;{\Gamma _1}\;\;\;\;\;{\Gamma _3} \hfill \\
  *\;\;\;\;\;\;\;\;*\;\;\;\;\;\; - \sigma I \hfill \\
\end{gathered}  \right] < 0\;.
\end{equation}
\end{lemma}

On the basis of Lemmas 1 and 2, the mixed ${H_2}/{H_\infty }$ performance will be analysed in this section, while the local estimation $\hat x_{k|k}^i$ satisfies Conditions 1-3 presented above.

\begin{theorem}\label{Theorem:1}
For a given level bound ${\gamma _i}$, the requirements of Conditions 1 and 2 are held. If and only if a matrix ${{\rm X}_i} > 0$ exists, the inequality in Eq.(\ref{eq:24}) is obtained as
\begin{shrinkeq}{-1.5ex}
\begin{equation}\label{eq:27}
\small
{\Delta _i} = \left[ \begin{gathered}
   - {{\rm X}_i}\;\;\;\;\;\;{{\rm X}_i}\left( {A_{t3}^i + \Delta A_{t3}^i} \right)\;\;\;\;\;\;{{\rm X}_i}B_{t3}^i \hfill \\
  *\;\;\;\;\;\;{\left( {D_{t3}^i} \right)^T}D_{t3}^i + {\Upsilon _{t3}} - {{\rm X}_i}\;\;\;\;\;0 \hfill \\
  *\;\;\;\;\;\;\;\;\;\;\;\;\;\;\;\;\;\;\;\;\;\;\;*\;\;\;\;\;\;\;\;\;\;\;\;\;\;\;\;\; - \gamma _i^2I \hfill \\
\end{gathered}  \right] < 0\;,
\end{equation}
\end{shrinkeq}
where ${\Upsilon _{t3}} \triangleq \sum\limits_{\vartheta  = 1}^\hbar  {{\varpi _{\vartheta ,k}}A_{\vartheta ,t3}^T{{\rm X}_i}{A_{\vartheta ,t3}}} $, and parameters $A_{t3}^i$, $\Delta A_{t3}^i$, ${A_{\vartheta ,t3}}$, $B_{t3}^i$, $G_{t3}^i$ as well as $D_{t3}^i$ are obtained using Eq.(\ref{eq:21}). In this case, the upper bound of ${H_2}$ performance can be calculated by
\begin{shrinkeq}{-1.5ex}
\begin{equation}\label{eq:28}
\small
{J^i} \leqslant Tr\left\{ {{R_k}{{\left( {G_{t3}^i} \right)}^T}{{\rm X}_i}G_{t3}^i} \right\}\;,
\end{equation}
\end{shrinkeq}
where the matrix ${{\rm X}_i}$ satisfies the inequality within Eq.(\ref{eq:27}), while $G_{t3}^i$ is denoted in Eq.(\ref{eq:21}) and ${R_k}$ is the covariance of ${v_k}$.
\end{theorem}
\begin{proof}[Proof]
The proof is derived and shown in Appendix C.
\end{proof}

Depending on the Gaussian white noise ${\varpi _{\vartheta ,k}}$ in Eq.(\ref{eq:1}) as well as ${\gamma _i}\;\left( {i = 1,\; \cdots \;,\;L} \right)$ given in Condition 2, the optimal filter parameters satisfying Conditions 1-3 are able to be determined from the solution of the following optimization problem:
\begin{shrinkeq}{-1.5ex}
\begin{equation}\label{eq:29}
\small
\min \;Tr\left( {{R_k}{{\left( {G_{t3}^i} \right)}^T}{{\rm X}_i}G_{t3}^i} \right)\;,
\end{equation}
\end{shrinkeq}
which is subject to the constraints of ${F_k}F_k^T \leqslant I$, Eqs.(\ref{eq:3}) and (\ref{eq:27}). However, solving the optimization problem from Eq.(\ref{eq:29}) is intractable, and the following three conditions should be analysed:

(i) For the objective function shown in Eq.(\ref{eq:29}), the matrix term ${\left( {G_{t3}^i} \right)^T}{{\rm X}_i}G_{t3}^i$ is nonlinear.

(ii) The uncertain parameters ${F_k}$ and ${\theta _{\vartheta ,k}}$ are included in the matrix inequality of Eq.(\ref{eq:27}).

(iii) ${{\rm X}_i}A_{t3}^i$ and ${{\rm X}_i}B_{t3}^i$ shown in Eq.(\ref{eq:27}) are nonlinear.

For the sake of obtaining an efficient solution of the optimization problem in Eq.(\ref{eq:29}), the convex optimization issue will be considered under a certain relaxation condition. The derivation is also shown in Appendix C.

\section{Distributed fusion estimation for logic ZOH}
For the stochastic uncertain system modelled in Eqs.(\ref{eq:1}) and (\ref{eq:16}), based on the robust finite horizon filtering, this paper studies a distributed fusion estimation approach (RFHDFE) by the aid of designing the local estimation $\hat x_{k|k}^i$. Then, a weighted robust fusion estimation is presented by investigating a convex optimization problem to improve the accuracy of the distributed estimation.

\subsection{Upper bound for estimation error}
To design the local estimation, it is necessary to examine the appropriate filter parameters, which are derived from the error covariance matrix. Since the proposed system model contains uncertain parameters and energy-bounded noise, it is difficult to obtain the accurate estimation for state. Therefore, the guaranteed upper bounds based on the augmentation method are presented by minimizing the estimation error. Then, Lemmas 3 and 4 for solving inequality constraints are introduced by transforming the upper bound of the error covariance.

\begin{lemma}\label{Lemma 3}
\cite{Bib22} Suppose that matrices $A$, $H$, $E$, $F$ and $X$ have certain dimensions, and satisfying $F{F^T} \leqslant I$. If there exists an arbitrary positive constant $\alpha  > 0$, and ${\alpha ^{ - 1}}I - EX{E^T} > 0$ is satisfied, which is a symmetric positive-definite matrix. Thus, the following inequality is held:
\begin{sequation}\label{eq:30}
\begin{gathered}
  \left( {A + HFE} \right)X{\left( {A + HFE} \right)^T} \hfill \\
   \leqslant A{\left( {{X^{ - 1}} - \alpha {E^T}E} \right)^{ - 1}}{A^T} + {\alpha ^{ - 1}}H{H^T}\;. \hfill \\
\end{gathered}
\end{sequation}

\noindent
Note that ${\left( {{X^{ - 1}} - \alpha {E^T}E} \right)^{ - 1}} = X + X{E^T}{\left( {{\alpha ^{ - 1}}I - EX{E^T}} \right)^{ - 1}}EX$ is transformed from the matrix inversion lemma.
\end{lemma}

\begin{lemma}\label{Lemma 4}
\cite{Bib22} Assuming that $X > 0$ and $Y > 0$ are symmetric positive definite matrices. At the time instant $k$, the functions satisfy the conditions ${s_t}\left( X \right) = s_t^T\left( X \right) \in {\mathbb{R}^{n \times n}}$ and ${h_t}\left( X \right) = h_t^T\left( X \right) \in {\mathbb{R}^{n \times n}}$. If there exists $Y > X$, such that the functions meet ${s_t}\left( Y \right) \geqslant {s_t}\left( X \right)$ and ${h_t}\left( Y \right) \geqslant {s_t}\left( Y \right)$, the solutions ${M_t}$ and ${N_t}$ are derived as follows:
\begin{sequation}\label{eq:31}
{M_{t + 1}} = {s_t}\left( {{M_t}} \right)\;,\;{N_{t + 1}} = {h_t}\left( {{N_t}} \right),\;{M_0} = {N_0} > 0
\end{sequation}
satisfying ${M_t} \leqslant {N_t}$.
\end{lemma}

\begin{theorem}\label{Theorem:2}
The unified form of $\left( {A + HFE} \right)X{\left( {A + HFE} \right)^T}$ from $\tilde \Theta _t^i$ in Eq.(B.4) and $\tilde \Sigma _{t + 1}^i$ in Eq.(B.5) based on Lemmas 3 and 4 satisfies the following case: if there exists a positive scalar ${\alpha _t} > 0$ and a symmetric positive-definite matrix $\Sigma _t^i$ both of which satisfy the inequality $\alpha _t^{ - 1}I - E_{t2}^i\Sigma _t^i{\left( {E_{t2}^i} \right)^T} > 0$, the upper bounds of the error $\tilde \Sigma _t^i \leqslant \Sigma _t^i$ and $\tilde \Theta _t^i \leqslant \Theta _t^i$ will be available. Then, the upper bounds of $\Theta _t^i$ and $\Sigma _{t + 1}^i$ are computed by the following recursive equations:
\begin{shrinkeq}{-2ex}
\begin{equation}\label{eq:32}
\small
\begin{gathered}
  \tilde \Theta _t^i \leqslant A_{t1}^i\Sigma _t^i{\left( {A_{t1}^i} \right)^T} + \;\alpha _t^{ - 1}H_{t1}^i{\left( {H_{t1}^i} \right)^T} + B_{t1}^iE\left( {{w_t}w_t^T} \right){\left( {B_{t1}^i} \right)^T} \hfill \\
  \;\; + A_{t1}^i\Sigma _t^i{\left( {E_{t1}^i} \right)^T}{\left( {\alpha _t^{ - 1}I - E_{t1}^i\Sigma _t^i{{\left( {E_{t1}^i} \right)}^T}} \right)^{ - 1}}E_{t1}^i\Sigma _t^i{\left( {A_{t1}^i} \right)^T} \hfill \\
  \;\; + G_{t1}^i{R_t}{\left( {G_{t1}^i} \right)^T} = \Theta _t^i\;, \hfill \\
\end{gathered}
\end{equation}
and
\begin{equation}\label{eq:33}
\small
\begin{gathered}
  \tilde \Sigma _{t + 1}^i \leqslant A_{t2}^i\Sigma _t^i{\left( {A_{t2}^i} \right)^T} + \alpha _t^{ - 1}H_{t2}^i{\left( {H_{t2}^i} \right)^T} + \Sigma _t^i{A_{\vartheta ,t2}}A_{\vartheta ,t2}^T{\theta _{\vartheta ,t}} \hfill \\
  \;\; + A_{t2}^i\Sigma _t^i{\left( {E_{t2}^i} \right)^T}{\left( {\alpha _t^{ - 1}I - E_{t2}^i\Sigma _t^i{{\left( {E_{t2}^i} \right)}^T}} \right)^{ - 1}}E_{t2}^i\Sigma _t^i{\left( {A_{t2}^i} \right)^T} \hfill \\
  \;\; + B_{t2}^iE\left( {{w_t}w_t^T} \right){\left( {B_{t2}^i} \right)^T} + G_{t2}^i{R_t}{\left( {G_{t2}^i} \right)^T} = \Sigma _{t + 1}^i\;. \hfill \\
\end{gathered}
\end{equation}
\end{shrinkeq}
\end{theorem}

\begin{proof}[Proof]
The derivation process is obtained from Lammas 3 and 4.
\end{proof}

Based on Theorem 2 and the finite horizon filtering, we define the error covariance matrix in the following form:
\begin{shrinkeq}{-1ex}
\begin{equation}\label{eq:34}
\small
\Sigma _t^i = \left[ \begin{gathered}
  \bar \Sigma _t^i\;\;\;\;\;\;\;\;0 \hfill \\
  0\;\;\;\;{P_t} - \bar \Sigma _t^i \hfill \\
\end{gathered}  \right]\;,
\end{equation}
\end{shrinkeq}
where $\bar \Sigma _t^i = E\left( {\tilde e_t^i{{\left( {\tilde e_t^i} \right)}^T}} \right)$ and ${P_t} = E\left( {{x_t}x_t^T} \right)$.

Then, the upper bounds of the error from the augmented state vectors given in Eqs.(B.4), (B.5), (\ref{eq:32}) and (\ref{eq:33}) are defined as follows:
\begin{shrinkeq}{-1ex}
\begin{equation}\label{eq:35}
\small
E\left( {e_t^i{{\left( {e_t^i} \right)}^T}} \right) = \left[ {I\;\;0} \right]\tilde \Theta _t^i\left[ \begin{gathered}
  I \hfill \\
  0 \hfill \\
\end{gathered}  \right] \leqslant \left[ {I\;\;0} \right]\Theta _t^i\left[ \begin{gathered}
  I \hfill \\
  0 \hfill \\
\end{gathered}  \right] = \bar \Theta _t^i
\end{equation}
and
\begin{equation}\label{eq:36}
\small
\begin{gathered}
  E\left( {\tilde e_{t + 1}^i{{\left( {\tilde e_{t + 1}^i} \right)}^T}} \right) = \left[ {I\;\;\;0} \right]\tilde \Sigma _{t + 1}^i\left[ \begin{gathered}
  I \hfill \\
  0 \hfill \\
\end{gathered}  \right] \hfill \\
   \leqslant \left[ {I\;\;\;0} \right]\Sigma _{t + 1}^i\left[ \begin{gathered}
  I \hfill \\
  0 \hfill \\
\end{gathered}  \right] = \bar \Sigma _{t + 1}^i\;. \hfill \\
\end{gathered}
\end{equation}
\end{shrinkeq}

Now, the ${H_\infty }$ performance of Condition 2 under Lemma 1 and Eqs.(\ref{eq:35})-(\ref{eq:36}) is derived as follows:
\begin{shrinkeq}{-1ex}
\begin{equation}\label{eq:37}
\small
\sum\limits_{k = 0}^\infty  {E\left( {{{\left( {\tilde e_k^i} \right)}^T}\tilde e_k^i} \right)}  \leqslant \sum\limits_{k = 0}^\infty  {\bar \Sigma _k^i}  < \gamma _i^2\sum\limits_{k = 0}^\infty  {E\left( {w_k^T{w_k}} \right)} \;.
\end{equation}
\end{shrinkeq}
Therefore, for a given ${H_\infty }$ DALB ${\gamma _i}$, the inequality in Eq.(\ref{eq:37}) holds, if and only if a matrix $\tilde \chi _{k + 1}^i$ for the estimation error using the augmented state vector $\tilde \Psi _k^i$ in Eq.(\ref{eq:19}) can be rewritten as:
\begin{shrinkeq}{-1ex}
\begin{equation}\label{eq:38}
\small
\begin{gathered}
  \tilde \chi _{k + 1}^i = \left( {A_{t4}^i + \Delta A_{t4}^i} \right)\tilde \chi _k^i + \left( {\sum\limits_{\vartheta  = 1}^\hbar  {{A_{\vartheta ,t4}}{\varpi _{\vartheta ,k}}} } \right)\tilde \chi _k^i \hfill \\
  \;\;{\text{ + }}B_{t4}^i{w_k}{\text{ + }}G_{t4}^i{v_k}\;, \hfill \\
  \tilde e_{k}^i = D_{t4}^i\tilde \chi _k^i\;, \hfill \\
\end{gathered}
\end{equation}
\end{shrinkeq}
where $\tilde \chi _{k}^i \triangleq \left[ \begin{gathered}
  {x_{k}} \hfill \\
  \tilde \Psi _{k}^i \hfill \\
\end{gathered}  \right]$, and the parameters are defined as:
\begin{shrinkeq}{-1ex}
\begin{equation}\label{eq:39}
\small
\begin{gathered}
  A_{t4}^i{\text{ = }}diag\left\{ {{A_k},\;A_{t2}^i} \right\}, \hfill \\
  \Delta A_{t4}^i{\text{ = }}diag\left\{ {{\mathcal{F}_k}{F_k}{E_k},\;H_{t2}^i{F_k}E_{t2}^i} \right\}, \hfill \\
  {A_{\vartheta ,t4}} = diag\left\{ {{A_\vartheta },\;{A_{\vartheta ,t2}}} \right\}\;, \hfill \\
  B_{t4}^i = col\left\{ {{B_k},\;B_{t2}^i} \right\},\;G_{t4}^i = col\left\{ {{G_k},\;G_{t2}^i} \right\}, \hfill \\
  D_{t4}^i = \left[ {0\;\;\;I\;\;\;0} \right]\;. \hfill \\
\end{gathered}
\end{equation}
\end{shrinkeq}
Note that $A_{t2}^i$, $H_{t2}^i$, $E_{t2}^i$, ${A_{\vartheta ,t2}}$, $B_{t2}^i$, $G_{t2}^i$ are defined in Eq.(B.3). Then, based on Theorem 1 and Eq.(C.5), the inequality is satisfied
\begin{shrinkeq}{-1ex}
\begin{equation}\label{eq:40}
\small
\left[ \begin{gathered}
   - {{\tilde {\rm X}}_i}\;\;\;\;\;\;{{\tilde {\rm X}}_i}\left( {A_{t4}^i + \Delta A_{t4}^i} \right)\;\;\;\;\;\;{{\tilde {\rm X}}_i}B_{t4}^i \hfill \\
  *\;\;\;\;\;\;{\left( {D_{t4}^i} \right)^T}D_{t4}^i + {\Upsilon _{t4}} - {{\tilde {\rm X}}_i}\;\;\;\;\;0 \hfill \\
  *\;\;\;\;\;\;\;\;\;\;\;\;\;\;\;\;\;\;\;\;\;\;\;*\;\;\;\;\;\;\;\;\;\;\;\;\;\;\;\;\; - \gamma _i^2I \hfill \\
\end{gathered}  \right] < 0\;,
\end{equation}
\end{shrinkeq}
where ${\Upsilon _{t4}} \triangleq \sum\limits_{\vartheta  = 1}^\hbar  {{\varpi _{\vartheta ,k}}A_{\vartheta ,t4}^T{{\tilde {\rm X}}_i}{A_{\vartheta ,t4}}} $.

On the other hand, the upper bound of the ${H_2}$ performance in Condition 3 is derived as follows:
\begin{shrinkeq}{-1.5ex}
\begin{equation}\label{eq:41}
\small
{J^i} \leqslant Tr\left\{ {{R_k}{{\left( {G_{t4}^i} \right)}^T}{{\tilde {\rm X}}_i}G_{t4}^i} \right\}\;.
\end{equation}
\end{shrinkeq}
The proof is similar to that of Theorem 1 and the inequality derivation of Eq.(C.7) from Eq.(C.5).

Next, in order to obtain $\bar \Sigma _t^i$ and ${P_t}$ in Eq.(\ref{eq:34}), Theorem \ref{Theorem:3} is presented depending on the optimal filter and the predictor, as well as the filter parameters for the $i^{th}$ local estimation shown in Eqs.(\ref{eq:17}) and (\ref{eq:18}).

\begin{theorem}\label{Theorem:3}
For the $i^{th}$ subsystem, the time-stamped measurement $z_t^i$ is transmitted using the logic ZOH. Meanwhile, the received data packet is denoted as $y_k^i$ with transmission delay ${\tau ^i}\left( {{k_1}} \right)$ at the time instant $k$. Set $t = k - {\tau ^i}\left( {{k_1}} \right)$, and ${\alpha _t} > 0$ be a positive scalar. Therefore, the matrices $\bar \Sigma _t^i$ and ${P_t}$ have the positive definite solutions, which are defined as follows:
\begin{shrinkeq}{-0.5ex}
\begin{equation}\label{eq:42}
\small
\bar \Theta _t^i = \bar \Sigma _t^i\left( {I + {{\left( {E_t^i} \right)}^T}{{\left( {\tilde M_t^i} \right)}^{ - 1}}E_t^i\bar \Sigma _t^i} \right) - \Lambda _t^i{\left( {\Xi _t^i} \right)^{ - 1}}{\left( {\Lambda _t^i} \right)^T}\;,
\end{equation}
\begin{equation}\label{eq:43}
\small
\bar \Sigma _{t + 1}^i = \tilde \Delta _t^i - L_t^i\tilde \nabla _t^i\;,
\end{equation}
\begin{equation}\label{eq:44}
\small
\begin{gathered}
  {P_{t + 1}} = {A_t}{\left( {P_t^{ - 1} - {\alpha _t}E_t^T{E_t}} \right)^{ - 1}}A_t^T + \alpha _t^{ - 1}{\mathcal{F}_t}\mathcal{F}_t^T \hfill \\
  \;\; + {\theta _{\vartheta ,k}}\sum\limits_{\vartheta  = 1}^\hbar  {{A_\vartheta }{P_t}A_\vartheta ^T}  + {B_t}E\left( {{w_t}w_t^T} \right)B_t^T + {G_t}{R_t}G_t^T\;. \hfill \\
\end{gathered}
\end{equation}
\end{shrinkeq}
In which
\begin{small}
\begin{shrinkeq}{-1ex}
\[\Lambda _t^i = \left( {I + \bar \Sigma _t^i{{\left( {E_t^i} \right)}^T}{{\left( {M_t^i} \right)}^{ - 1}}E_t^i} \right)\bar \Sigma _t^i{\left( {C_t^i} \right)^T}\;,\]
\[\begin{gathered}
  \Xi _t^i = C_t^i\bar \Sigma _t^i\;\left( {I + {{\left( {E_t^i} \right)}^T}{{\left( {M_t^i} \right)}^{ - 1}}E_t^i\bar \Sigma _t^i} \right){\left( {C_t^i} \right)^T} \hfill \\
  \;\; + \alpha _t^{ - 1}\mathcal{H}_t^i{\left( {\mathcal{H}_t^i} \right)^T} + B_t^iE\left( {{w_t}w_t^T} \right){\left( {B_t^i} \right)^T} + G_t^i{R_t}{\left( {G_t^i} \right)^T}\;, \hfill \\
\end{gathered} \]

\[\begin{gathered}
  \Delta _t^i = {A_t}\left( {I + \bar \Sigma _t^i{{\left( {E_t^i} \right)}^T}{{\left( {M_t^i} \right)}^{ - 1}}E_t^i} \right)\bar \Sigma _t^i{\left( {C_t^i} \right)^T} \hfill \\
  \;\; + \alpha _t^{ - 1}{\mathcal{F}_t}{\left( {\mathcal{H}_t^i} \right)^T} + {B_t}E\left( {{w_t}w_t^T} \right){\left( {B_t^i} \right)^T} + {G_t}{R_t}{\left( {G_t^i} \right)^T}\;, \hfill \\
\end{gathered} \]

\[\begin{gathered}
  \nabla _t^i = C_t^i\left( {I + \bar \Sigma _t^i{{\left( {E_t^i} \right)}^T}{{\left( {\tilde M_t^i} \right)}^{ - 1}}E_t^i} \right)\bar \Sigma _t^i{\left( {C_t^i} \right)^T} \hfill \\
  \;\; + \alpha _t^{ - 1}\mathcal{H}_t^i{\left( {\mathcal{H}_t^i} \right)^T} + B_t^iE\left( {{w_t}w_t^T} \right){\left( {B_t^i} \right)^T} + G_t^i{R_t}{\left( {G_t^i} \right)^T}\;, \hfill \\
\end{gathered} \]

\[\begin{gathered}
  \tilde \Delta _t^i = {A_t}\left( {I + \bar \Sigma _t^i{{\left( {E_t^i} \right)}^T}{{\left( {M_t^i} \right)}^{ - 1}}E_t^i} \right)\bar \Sigma _t^iA_t^T + \alpha _t^{ - 1}{\mathcal{F}_t}\mathcal{F}_t^T \hfill \\
  \;\; + {B_t}E\left( {{w_t}w_t^T} \right)B_t^T + {G_t}{R_t}G_t^T + \sum\limits_{\vartheta  = 1}^\hbar  {{A_\vartheta }} {P_t}\left( {A_\vartheta ^T} \right){\theta _{\vartheta ,t}}\;, \hfill \\
\end{gathered} \]
\end{shrinkeq}
and
\[\begin{gathered}
  \tilde \nabla _t^i = C_t^i\left( {I + \bar \Sigma _t^i{{\left( {E_t^i} \right)}^T}{{\left( {M_t^i} \right)}^{ - 1}}E_t^i} \right)\bar \Sigma _t^iA_t^T \hfill \\
  \;\; + \alpha _t^{ - 1}\mathcal{H}_t^i\mathcal{F}_t^T + B_t^iE\left( {{w_t}w_t^T} \right)B_t^T + G_t^i{R_t}G_t^T\;. \hfill \\
\end{gathered} \]
Meanwhile, it satisfies $P_t^{ - 1} - {\alpha _t}{\left( {E_t^i} \right)^T}E_t^i > 0$, $\tilde M_t^i = \alpha _t^{ - 1}I - E_t^i{P_t}{\left( {E_t^i} \right)^T}$ and $M_t^i = \alpha _t^{ - 1}I - E_t^i\bar \Sigma _t^i{\left( {E_t^i} \right)^T} > 0$, respectively.
\end{small}

Furthermore, the re-organized local estimation given in Eqs.(\ref{eq:17})-(\ref{eq:18}) is facilitated by using the following filter parameters:
\begin{shrinkeq}{-1ex}
\begin{equation}\label{eq:45}
\small
\hat C_t^i = C_t^i\left( {I + \bar \Sigma _t^i{{\left( {E_t^i} \right)}^T}{{\left( {M_t^i} \right)}^{ - 1}}E_t^i} \right)\;,
\end{equation}
\begin{equation}\label{eq:46}
\small
K_t^i = \Lambda _t^i{\left( {\Xi _t^i} \right)^{ - 1}}\;,
\end{equation}
\begin{equation}\label{eq:47}
\small
\hat A_t^i = {A_t}\left( {I + \bar \Sigma _t^i{{\left( {E_t^i} \right)}^T}{{\left( {M_t^i} \right)}^{ - 1}}E_t^i} \right)\;,
\end{equation}
\begin{equation}\label{eq:48}
\small
L_t^i = \Delta _t^i{\left( {\nabla _t^i} \right)^{ - 1}}\;.
\end{equation}
\end{shrinkeq}
\end{theorem}

\begin{proof}[Proof]
This theorem comes from the upper bounds for minimizing the estimation error in the covariance matrices. Appendix D presents the procedure of proof.
\end{proof}

\subsection{Linear compensation for packet dropouts}
To solve the local state estimation $\hat x_{k|k}^i$, this section applies the optimal estimation $\hat x_{t|t}^i$ to deriving the state $\hat x_{k|t}^i$ at the current time instant $k$, which is obtained from the maximum of the transmission delay $N$ as well as the current received valid signal. So the predictor $\hat x_{t + 1|t}^i$ is introduced to solve the linear compensation, that is:
\begin{shrinkeq}{-1ex}
\begin{equation}\label{eq:49}
\small
\hat x_{k|t}^i = \left( {1 - \frac{{{\tau ^i}\left( {{k}} \right) - 1}}{N}} \right)\hat x_{t + 1|t}^i\;,
\end{equation}
\end{shrinkeq}
where ${\tau ^i}\left( k \right)$ denotes the transmission delay of the received valid signal. The acknowledgment (ACK) time-stamped signal $z_t^i$ is assigned the highest transmission priority due to the timing, whilst the transmission delay is negligible to obtain the estimated state $\hat x_{k|t}^i$ \cite{Bib38}. Thus, the filter $\hat x_{t|t}^i$ and predictor $\hat x_{t + 1|t}^i$ are used for compensating the estimated state $\hat x_{k|t}^i$ given in Eq.(\ref{eq:49}).

\begin{remark}\label{Remark 5}
Note that it is an approximate estimation using the linear compensation method for the estimated state $\hat x_{k|t}^i$ with delay-free in Eq.(\ref{eq:49}). The filter $\hat x_{t|t}^i$ derived from the latest data packet with time-stamp is used to estimate the current state $\hat x_{k|k}^i$ from the compensated value $\hat x_{k|t}^i$. To deal with the packet dropouts, the data transmission delays for each step will be ignored until the next received data packet. Due to the growing error accumulation and computational burden using the traditional one-step predication, we propose the linear compensation for packet dropouts, which has the advantage of simplifying the state estimation and highlighting the valid transmission signals. Therefore, the proposed method is able to reduce the computational complexity, and then alleviate the negative effect owing to the random transmission delays. More importantly, define $s = k + 1 - {\tau ^i}\left( {k + 1} \right)$ with transmission delay ${\tau ^i}\left( {k + 1} \right)$ at the next sampling time $k+1$, when $s > t + 1$, the missing data packets are produced due to the signal selection scheme. Finally, the estimated state $\hat x_{s|s}^i$ is calculated by using the one-step prediction, which meets the artificial delay $1 < {\tau ^{st}}\left( k \right) = s - t \leqslant N$. For the missing data packets, the re-organized state estimation sequence will be compensated, and the estimated values are listed as follows:
\begin{shrinkeq}{-1ex}
\begin{equation}\label{eq:50}
\left\{ {\hat x_{t + 1|t + 1}^i,\; \cdots \;,\;\hat x_{t + {\tau ^{st}}\left( k \right)|t + {\tau ^{st}}\left( k \right)}^i} \right\}\;.
\end{equation}
\end{shrinkeq}
\noindent
Therefore, the objective of the linear compensation with the random transmission delays, as well as the one-step prediction for the missing packets, is designed the fusion estimation of each local subsystem subsequently.
\end{remark}

\subsection{Distributed fusion estimation based on local attributes}
Theorem \ref{Theorem:3} presents a local state estimation, meanwhile, the current state is estimated by using the proposed linear compensation for packet dropouts. In addition, to improve the state estimation accuracy from each subsystem, this section investigates the distributed fusion estimation method with the aid of the weighted fusion criterion.

\begin{lemma}\label{Lemma 5}
For the linear discrete-time systems described with Eqs.(\ref{eq:1}) and (\ref{eq:2}) at time instant $k$, based on the established signal selection model and linear compensation for packet dropouts, the upper bounds for cross-covariance matrix $\bar \Theta _t^{i,j}$ of the filtering error and cross-covariance matrix $\bar \Sigma _{t + 1}^{i,j}$ for the prediction error between the $i^{th}$ and $j^{th}$ subsystems have the following recursive expression:
\begin{shrinkeq}{-2ex}
\begin{sequation}\label{eq:51}
\begin{gathered}
  \bar \Theta _t^{i,j} = \left( {I - K_t^iC_t^i} \right)\bar \Sigma _t^{i,j}\;{\left( {I - K_t^jC_t^j} \right)^T} \hfill \\
  \;\; + K_t^i\left( {\hat C_t^i - C_t^i} \right)\left( {{P_t} - \bar \Sigma _t^{i,j}} \right){\left( {K_t^j\left( {\hat C_t^j - C_t^j} \right)} \right)^T} \hfill \\
  \;\; + \;\alpha _t^{ - 1}K_t^i\mathcal{H}_t^i{\left( {K_t^j\mathcal{H}_t^j} \right)^T} + K_t^iG_t^i{R_t}{\left( {K_t^jG_t^j} \right)^T} \hfill \\
  \;\; + K_t^iB_t^iE\left( {{w_t}w_t^T} \right){\left( {K_t^jB_t^j} \right)^T} \hfill \\
  \;\; + \left( {\bar \Sigma _t^{i,j} + K_t^i\left( {\hat C_t^i\left( {{P_t} - \bar \Sigma _t^{i,j}} \right) - C_t^i{P_t}} \right)} \right) \hfill \\
  \;\; \times {\left( {E_t^i\;} \right)^T}{\left( {\tilde M_t^{i,j}} \right)^{ - 1}}E_t^j \hfill \\
  \;\; \times {\left( {\bar \Sigma _t^{i,j} + K_t^j\left( {\hat C_t^j\left( {{P_t} - \bar \Sigma _t^{i,j}} \right) - C_t^j{P_t}} \right)} \right)^T}\;, \hfill \\
\end{gathered}
\end{sequation}

and
\[\begin{gathered}
  \bar \Sigma _{t + 1}^{i,j} = \left( {{A_t} - L_t^iC_t^i} \right)\bar \Sigma _t^{i,j}{\left( {{A_t} - L_t^jC_t^j} \right)^T} \hfill \\
  \;\; + \left( {{A_t} - \hat A_t^i + L_t^i\left( {\hat C_t^i - C_t^i} \right)} \right)\left( {{P_t} - \bar \Sigma _t^{i,j}} \right) \hfill \\
  \;\; \times {\left( {{A_t} - \hat A_t^j + L_t^j\left( {\hat C_t^j - C_t^j} \right)} \right)^T} \hfill \\
  \;\; + \alpha _t^{ - 1}\left( {{\mathcal{F}_t} - L_t^i\mathcal{H}_t^i} \right){\left( {{\mathcal{F}_t} - L_t^j\mathcal{H}_t^j} \right)^T} \hfill \\
  \;\; + \sum\limits_{\vartheta  = 1}^\hbar  {{A_\vartheta }} {P_t}\left( {A_\vartheta ^T} \right){\theta _{\vartheta ,t}} \hfill \\
\end{gathered} \]

\begin{sequation}\label{eq:52}
\begin{gathered}
  \;\; + \left( {{B_t} - L_t^iB_t^i} \right)E\left( {{w_t}w_t^T} \right){\left( {{B_t} - L_t^jB_t^j} \right)^T} \hfill \\
  \;\; + \left( {{G_t} - L_t^iG_t^i} \right){R_t}{\left( {{G_t} - L_t^jG_t^j} \right)^T} \hfill \\
  \;\; + \left( {\left( {{A_t} - L_t^iC_t^i} \right)\bar \Sigma _t^{i,j}} \right. \hfill \\
  \;\;\;\;\left. { + \left( {{A_t} - \hat A_t^i + L_t^i\left( {\hat C_t^i - C_t^i} \right)} \right)\left( {{P_t} - \bar \Sigma _t^{i,j}} \right)} \right) \hfill \\
  \;\; \times {\left( {E_t^i} \right)^T}{\left( {\tilde M_t^{i,j}} \right)^{ - 1}}E_t^j \hfill \\
  \;\; \times \left( {\left( {{A_t} - L_t^iC_t^i} \right)\bar \Sigma _t^{i,j}} \right. \hfill \\
  \;\;\;\;{\left. { + \left( {{A_t} - \hat A_t^j + L_t^j\left( {\hat C_t^j - C_t^j} \right)} \right)\left( {{P_t} - \bar \Sigma _t^{i,j}} \right)} \right)^T}\; \hfill \\
\end{gathered}
\end{sequation}
\end{shrinkeq}

\noindent
where $\tilde M_t^{i,j} = \alpha _t^{ - 1}I - E_t^i{P_t}{\left( {E_t^j} \right)^T}$ and $M_t^{i,j} = \alpha _t^{ - 1}I - E_t^i\bar \Sigma _t^{i,j}{\left( {E_t^j} \right)^T}$.
\end{lemma}

To obtain the distributed fusion estimation ${\hat x_{k|k}}$, the considered system develops the linear minimum covariance, as well as the optimal weighted fusion from the local estimation. Therefore, the designed estimator from Eq.(\ref{eq:1}) is described as the following convex optimization form:
\begin{shrinkeq}{-1.5ex}
\begin{equation}\label{eq:53}
\small
{\hat x_{k|k}} = \sum\limits_{i = 1}^L {\Omega _k^i\hat x_{k|k}^i} \;.
\end{equation}
\end{shrinkeq}

To have the optimal value for the distributed fusion estimation \cite{Bib11,Bib40,Bib39_2}, if and only if the weighted matrix in Eq.(\ref{eq:53}) is composed by
\begin{shrinkeq}{-1.5ex}
\begin{equation}\label{eq:54}
\small
{\Omega _k} = \left[ {\Omega _k^1,\; \cdots ,\;\Omega _k^L} \right] = {\left( {I_0^T\Pi _k^{ - 1}{I_0}} \right)^{ - 1}}I_0^T\Pi _k^{ - 1}\;,
\end{equation}
\end{shrinkeq}
where ${I_0} = {\underbrace {\left[ {{I_r}\;,\; \cdots ,\;{I_r}} \right]}_L}^T$ is an $rL \times r$-dimensional matrix, and $\sum\limits_{i = 1}^L {\Omega _k^i}  = {I_r}$. Moreover, the covariance matrix based on the linear compensation is defined as $\Pi _k^{i,j} \triangleq E\left( {e_k^i{{\left( {e_k^j} \right)}^T}} \right)$, and
\begin{shrinkeq}{-1.5ex}
\begin{equation}\label{eq:55}
\small
{\Pi _k} = \left[ \begin{gathered}
  \Pi _k^{1,1}\;\; \cdots \;\;\Pi _k^{1,L} \hfill \\
  \;\;\;\;\;\;\;\; \ddots \;\;\; \hfill \\
  \Pi _k^{L,1}\;\; \cdots \;\;\Pi _k^{L,L} \hfill \\
\end{gathered}  \right] = {\left( {\Pi _k^{i,j}} \right)_{rL \times rL}}
\end{equation}
\end{shrinkeq}
is a symmetrical positive-definite matrix. Set ${\tilde \Pi _k} = {\left( {\Pi _k^{{{ - 1} \mathord{\left/
 {\vphantom {{ - 1} 2}} \right.
 \kern-\nulldelimiterspace} 2}}{I_0}} \right)^T}\left( {\Pi _k^{{1 \mathord{\left/
 {\vphantom {1 2}} \right.
 \kern-\nulldelimiterspace} 2}}I_0^i} \right)$ and $I_0^i = {\underbrace {\left[ {0, \cdots ,{I_r}, \cdots ,0} \right]}_L^T}$ be the $i^{th}$ element in ${I_0}$. It is worth noting that on the basis of the cross-covariance of ${\tilde \Pi _{k|k}} \triangleq E\left( {\left( {{x_k} - {{\hat x}_{k|k}}} \right){{\left( {{x_k} - {{\hat x}_{k|k}}} \right)}^T}} \right)$, the corresponding optimal fusion estimation is designed by ${\tilde \Pi _{k|k}} = {\left( {I_0^T\Pi _k^{ - 1}{I_0}} \right)^{ - 1}}$, which satisfies ${\tilde \Pi _{k|k}} \leqslant \Pi _k^i \leqslant \bar \Theta _k^i$ under Eq.(\ref{eq:35}), if $i = j$; otherwise, we have ${\tilde \Pi _{k|k}} = \Pi _k^i \leqslant \bar \Theta _k^i$ by setting $\Omega _k^i = {I_r}$ and $\Omega _k^j = 0$ ($j = 1,\; \cdots \;,\;L,\;j \ne i$).

\begin{remark}\label{Remark 6}
The stochastic uncertain system presents the signal selection scheme to handle network-induced complexity simultaneously. Meanwhile, since the received data packets include the latest signals, the packet disorders are actively avoided. That is, a linear compensation method for the packet dropouts is obtained from each local estimation $\hat x_{k|k}^i$. Furthermore, to achieve information exchange between any two subsystems, we propose the convex optimization problem based on the weighted fusion criterion, and the optimal estimated state is available by using the linear minimum covariance matrix. Thus, the fusion estimation possesses higher estimation accuracy than each local subsystem.
\end{remark}

Summarizing the throughout calculation process of the distributed fusion estimation, the solution approach is illustrated in Fig.3.

\begin{figure*}[t]
	\begin{center}
		\includegraphics[width=0.5\linewidth]{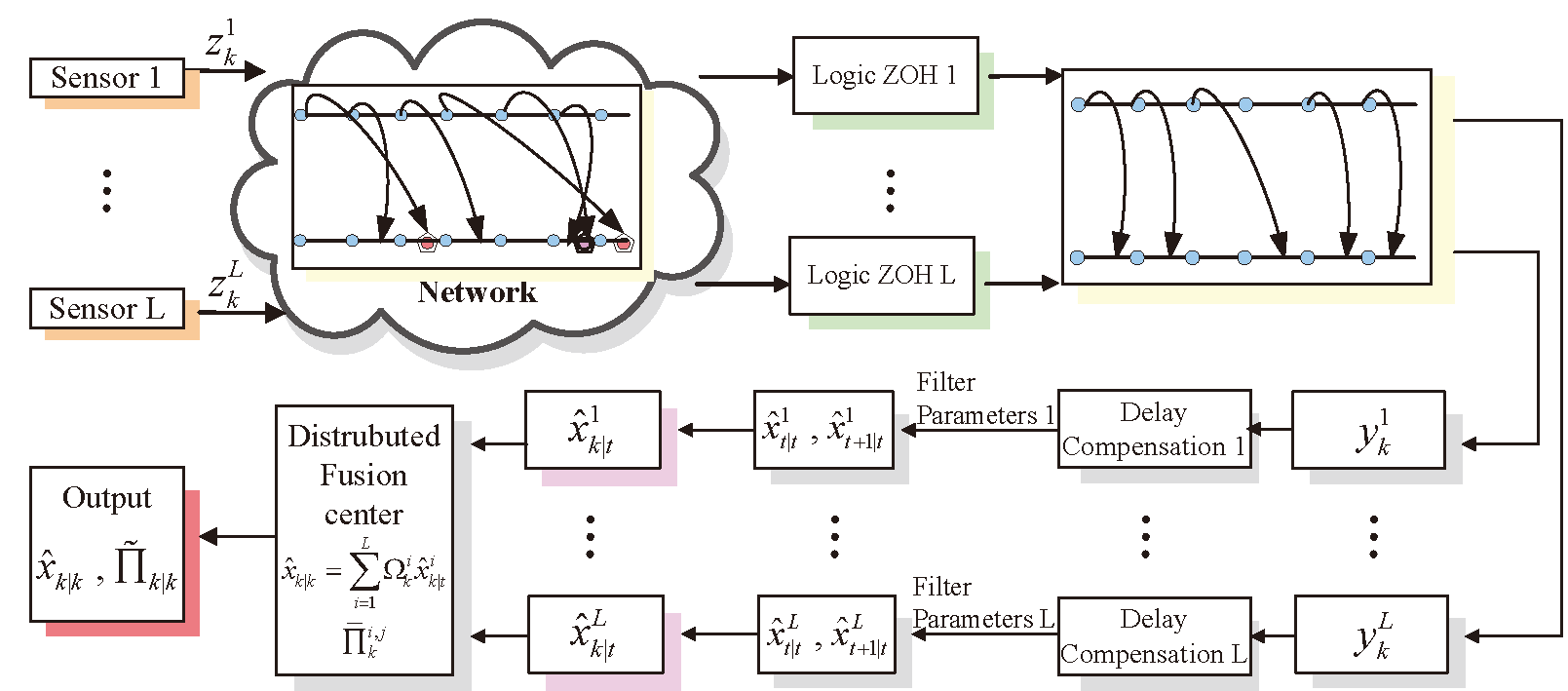}
	\end{center}
    \vspace{-4mm}
	\caption{Computational procedure of RFHDFE}
	\label{Fig 3}
\end{figure*}

\subsection{Stability analysis}
To analysis the steady-state local estimators, the definition and formalism of Lemma 6 is shown.

\begin{lemma}\label{Lemma 6}
\cite{Bib24,Bib40,Bib47} For the re-organized system shown in Eqs.(\ref{eq:1}) and (\ref{eq:16}), the ${H_2}$ performance is represented in Eq.(\ref{eq:41}), and ${\tilde {\rm X}_i} \leqslant {{\rm X}_i}$ is satisfied under Eq.(\ref{eq:40}). Given the covariance matrix $E\left( {\tilde e_{k + 1}^i{{\left( {\tilde e_{k + 1}^i} \right)}^T}} \right)$ by Eq.(\ref{eq:36}) with any initial value $E\left( {\tilde e_0^i{{\left( {\tilde e_0^i} \right)}^T}} \right) \geqslant 0$ due to the external disturbance, the error covariance will converge to the positive semi-definite solution ${\overset{\lower0.5em\hbox{$\smash{\scriptscriptstyle\smile}$}}{\Sigma^i}}$. According to the upper bound for the covariance matrix of the estimation error, the inequality $E\left( {\tilde e_{k + 1}^i{{\left( {\tilde e_{k + 1}^i} \right)}^T}} \right) \leqslant \bar \Sigma _{k + 1}^i$ is satisfied due to Eq.(\ref{eq:43}), as well as $Tr\left( {E\left( {\tilde e_{k + 1}^i{{\left( {\tilde e_{k + 1}^i} \right)}^T}} \right)} \right) \leqslant Tr\left( {\bar \Sigma _{k + 1}^i} \right)$, which is rewritten as follows:
\begin{sequation}\label{eq:56}
\begin{gathered}
  \bar \Sigma _{k + 1}^i = \left( {{A_k} - L_k^iC_k^i} \right)\left( {I + \bar \Sigma _k^i{{\left( {E_k^i} \right)}^T}{{\left( {M_k^i} \right)}^{ - 1}}E_k^i} \right)\bar \Sigma _k^iA_k^T \hfill \\
  \;\; + \alpha _k^{ - 1}\left( {{\mathcal{F}_k} - L_k^i\mathcal{H}_k^i} \right)\mathcal{F}_k^T \hfill \\
  \;\; + \left( {{B_k} - L_k^iB_k^i} \right)E\left( {{w_k}w_k^T} \right)B_k^T \hfill \\
  \;\; + \left( {{G_k} - L_k^iG_k^i} \right){R_k}G_k^T + \sum\limits_{\vartheta  = 1}^\hbar  {{A_\vartheta }} {P_k}\left( {A_\vartheta ^T} \right){\theta _{\vartheta ,k}}\;. \hfill \\
\end{gathered}
\end{sequation}
Let the set filter parameters be ${\overset{\lower0.5em\hbox{$\smash{\scriptscriptstyle\smile}$}}{L^i}} = {\lim _{k \to \infty }}L_k^i$ given in Eq.(\ref{eq:48}), ${\overset{\lower0.5em\hbox{$\smash{\scriptscriptstyle\smile}$}}{\Delta ^i}} = {\lim _{k \to \infty }}\tilde \Delta _k^i$ and ${\overset{\lower0.5em\hbox{$\smash{\scriptscriptstyle\smile}$}}{\nabla ^i}} = {\lim _{k \to \infty }}\tilde \nabla _k^i$ in Eq.(D.9). Therefore, ${\overset{\lower0.5em\hbox{$\smash{\scriptscriptstyle\smile}$}}{\Sigma ^i}} = {\lim _{k \to \infty }}E\left( {\tilde e_{k + 1}^i{{\left( {\tilde e_{k + 1}^i} \right)}^T}} \right)$ is determined by $\bar \Sigma _k^i$. We know that matrix ${A_k}$ is stable, and $\rho \left( {{A_k}} \right)$ denotes the spectrum radius of matrix ${A_k}$. In this situation, the estimated state $\hat x_{k + 1|k}^i$ in Eq.(\ref{eq:18}) is asymptotically stable.
\end{lemma}

Note that the stability of matrix ${A_k}$ for the steady-state filter is necessary. In addition, the error covariance matrix $E\left( {\tilde e_{k + 1}^i{{\left( {\tilde e_{k + 1}^i} \right)}^T}} \right)$ is influenced by energy-bounded noise and uncertain white noise. If matrix ${A_k}$ is unstable, the error covariance matrices of the estimation errors will be unbounded. Meanwhile, the covariance matrices of the filtering errors are solved by $E\left( {e_k^i{{\left( {e_k^i} \right)}^T}} \right) \leqslant \bar \Theta _k^i$ shown in Eq.(\ref{eq:35}), and $\bar \Theta _k^i$ is determined by $\bar \Sigma _k^i$ according to Eq.(\ref{eq:42}). Therefore, the steady-state cross-covariance matrices for the filtering errors are converging.

Finally, the proposed distributed fusion estimation based on local estimation is the steady-state filter under Eqs.(\ref{eq:53}) and (\ref{eq:54}). Note that $\hat x_{k|k}^i = \left[ {0\;\;{I_r}} \right]\Psi _k^i$ in Eq.(\ref{eq:19}) is the local steady-state filter for the state ${x_k}$, and ${\Omega _k}$ is the steady-state weighted matrix. Moreover, ${\Pi _k}$ in Eq.(\ref{eq:55}) is the steady-state covariance matrix for the filtering error. Therefore, the variance matrix ${\left( {I_0^T\Pi _k^{ - 1}{I_0}} \right)^{ - 1}}$ is the distributed fusion estimator.

\begin{remark}\label{Remark 7}
The proposed distributed fusion estimation (i.e. RFHDFE) approach is not required to be calculated at each step, so that it is easy to be implemented in engineering application. In addition, the proposed RFHDFE alleviates the computational burden. Since it has the computational cost with the order of magnitude $O\left( {\sum\limits_{i = 1}^L {{{\left( {\frac{{{\tau ^i}\left( k \right)}}{N}\left( {{m_i} + r} \right)} \right)}^2}} } \right)$, where $r$ expresses the dimension of the state, ${m_i}$ represents the measurement dimension and $L$ denotes the number of the subsystems. Therefore, this paper investigates the distributed fusion estimation based on the local estimation possesses less computation complexity and reduces the energy consumption.
\end{remark}

\section{Numerical simulation}
We use simulations to illustrate the effectiveness of the proposed RFHDFE approach. Considering that a moving target tracking system is measured by three sensors, which is described by the following state-space and measurement models, respectively, and the related systems are modeled in \cite{Bib11,Bib22,Bib34}:
\begin{shrinkeq}{-1ex}
\begin{equation}\label{eq:57}
\small
\begin{gathered}
  {x_{k + 1}} = \left( {{A_k} + {\mathcal{F}_k}{F_k}{E_k} + {A_1}{\varpi _{1,k}}} \right){x_k} \hfill \\
  \;\; + {B_k}{w_k} + {G_k}{v_k},\;\;\;\;\;\;\;\;k = 1,2, \cdots  \hfill \\
\end{gathered}
\end{equation}
\begin{equation}\label{eq:58}
\small
z_k^i = \left( {C_k^i + \mathcal{H}_k^i{F_k}E_k^i} \right){x_k} + B_k^i{w_k} + G_k^i{v_k}\;,\;\;\;\;\;i = 1,\;2,\;3
\end{equation}
\end{shrinkeq}
where the system parameters are set as follows:
\begin{shrinkeq}{-1.5ex}
\begin{small}
\[\begin{gathered}
  {A_k} = \left[ \begin{gathered}
  0.9\;\;\;\;T\;\;\;\;{{{T^2}} \mathord{\left/
 {\vphantom {{{T^2}} 2}} \right.
 \kern-\nulldelimiterspace} 2} \hfill \\
  0\;\;\;\;\;0.9\;\;\;\;\;T \hfill \\
  0\;\;\;\;\;\;0\;\;\;\;\;\;0.9 \hfill \\
\end{gathered}  \right],\;{A_1} = \left[ \begin{gathered}
  0.02\;\;\;0.03\;\;\;0.01 \hfill \\
  0.06\;\;\;0.05\;\;\;0.02 \hfill \\
  0.05\;\;\;0.03\;\;\;0.01 \hfill \\
\end{gathered}  \right], \hfill \\
  {B_k} = \left[ \begin{gathered}
  0.1 \hfill \\
  0.3 \hfill \\
  0.2 \hfill \\
\end{gathered}  \right],\;B_k^1 = \left[ \begin{gathered}
  0.1 \hfill \\
  0.3 \hfill \\
  0.2 \hfill \\
\end{gathered}  \right],\;B_k^2 = \left[ \begin{gathered}
  0.2 \hfill \\
  0.4 \hfill \\
  0.5 \hfill \\
\end{gathered}  \right],B_k^3 = \left[ \begin{gathered}
  0.4 \hfill \\
  0.5 \hfill \\
  0.2 \hfill \\
\end{gathered}  \right] \hfill \\
  {G_k} = col\left\{ {{{{T^2}} \mathord{\left/
 {\vphantom {{{T^2}} 2}} \right.
 \kern-\nulldelimiterspace} 2},\;T,1} \right\},\;G_k^i = B_k^i\left( {i = 1,\;2,\;3} \right), \hfill \\
  {\mathcal{F}_k} = {\left[ {0.1\;\;0.1\;\;0.1} \right]^T},\;{E_k} = \left[ {0.02\;\;0.02\;\;0.02} \right], \hfill \\
  C_k^1 = \left[ {0.6\;\;\;0.8\;\;\;1} \right],\;C_k^2 = \left[ {1\;\;\;0.8\;\;\;0.5} \right]\;, \hfill \\
  C_k^3 = \left[ {0.3\;\;\;1\;\;\;0.7} \right],\;\mathcal{H}_k^1 = \mathcal{H}_k^2 = \mathcal{H}_k^3 = 0.8\;. \hfill \\
\end{gathered} \]
\end{small}

\noindent
Symbol $T$ denotes the sample period and is set to be $0.1s$. The defined state ${x_k} \triangleq {\left( {{s_k}\;\;\;{{\dot s}_k}\;\;\;{{\ddot s}_k}} \right)^T}$ includes the position, velocity and acceleration, respectively, of the moving target at the time instant $kT$, and the maximum of the transmission delay is set $N = 5$. The parameter uncertainty ${F_k} = \sin \left( {0.6k} \right)$ is a time-varying matrix. ${v_k}$ and ${\varpi _{1,k}}$ are the Gaussian white noise with covariance ${R_k} = 0.09$ and $\theta _1^L \leqslant {\theta _{1,k}} \leqslant \theta _1^U$, where the covariance satisfies $\theta _1^L = 0.01$ and $\theta _1^U = 0.02$. The DALB ${\gamma _1} = {\gamma _2} = {\gamma _3} = 0.2$, and the energy bounded noise ${w_k}$ is given by ${w_k} = 2\cos \left( {0.6k} \right)$.
\end{shrinkeq}

Without losing generality, to solve Theorem \ref{Theorem:3} with appropriate filter parameters, the initial values are set as ${\hat x_{0|0}} = {\mu _0} = E\left( {{x_0}} \right) = {\left[ {1\;\;\;\;1\;\;\;\;1} \right]^T}$, ${P_{0|0}} = 0.01{I_3}$ as well as ${\alpha _k} = 3$. Fig.4 demonstrates the traces for the covariance matrices of the estimation errors. The estimated state ${\hat x_{k|k}}$ using the proposed RFHDFE with delay-free is compared against the improved robust finite-horizon Kalman filtering (IRFHKF) approach \cite{Bib22}. From Fig.4 (a), the proposed RFHDFE approach has a less upper bound than IRFHKF for the covariance matrices of the estimation errors. Note that due to the weighted fusion criterion, the fused upper bound for the cross-covariance matrix of the error is much smaller than that of the other covariance matrices, and the fusion estimation criterion is optimal. Comparing the estimation accuracy of the both methods shown in Figs.4 (b)-(d), it implies that the proposed RFHDFE method results in an appropriate estimator, which is suitable for deriving the optimal filter parameters and obtaining the precise estimation. Taking the energy bound noise and unknown state-dependent noise into account, the dynamic tracking trajectory is very close to the actual one.

The comparison of the mean square error values (MSEs) \cite{Bib11,Bib40} are provided by Monte-Carlo shown in Table 1, to indicate the accuracy improvement of the proposed RFHDFE approach with network-induced complexity and multiple noise in different noise environments. Note that the comparison of MSEs includes the states of position, velocity and acceleration, as well as their estimated values and runtime for each method, respectively. It shows that the MSEs of the fusion estimation are smaller than each subsystem, and also smaller than the IRFHKF method. It means that the RFHDFE approach has the capacity for obtaining better estimation performance than each one and the IRFHKF method in the following aspects: smaller MSEs and less time consuming.

\begin{shrinkeq}{-1ex}
\begin{figure}[H]
\setcounter{subfigure}{0}
{\centering
\subfigure[Upper bound for trace of error covariance matrices]{\includegraphics[width=0.8\linewidth]{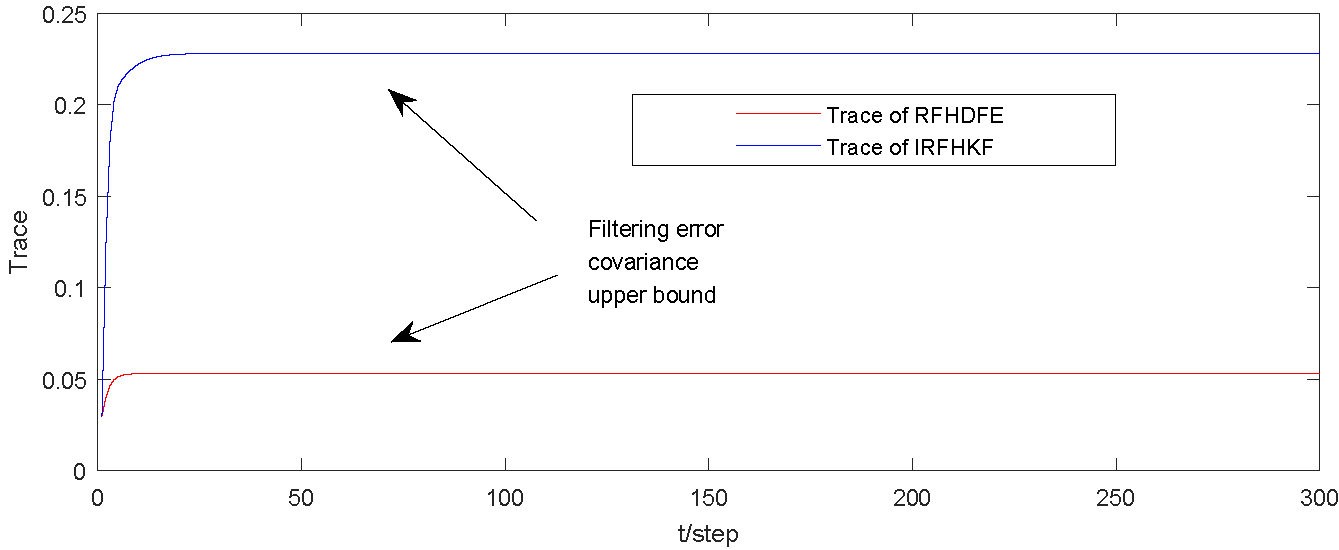}}
\\
\subfigure[State estimation of position]{\includegraphics[width=0.8\linewidth]{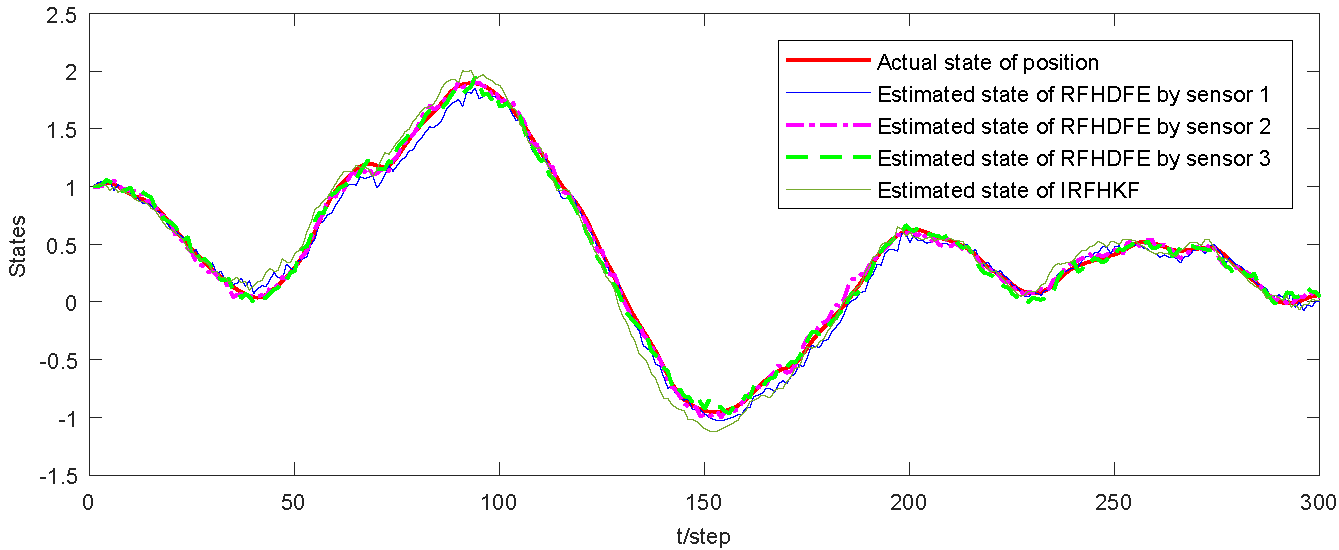}}
\\
\subfigure[State estimation of velocity]{\includegraphics[width=0.8\linewidth]{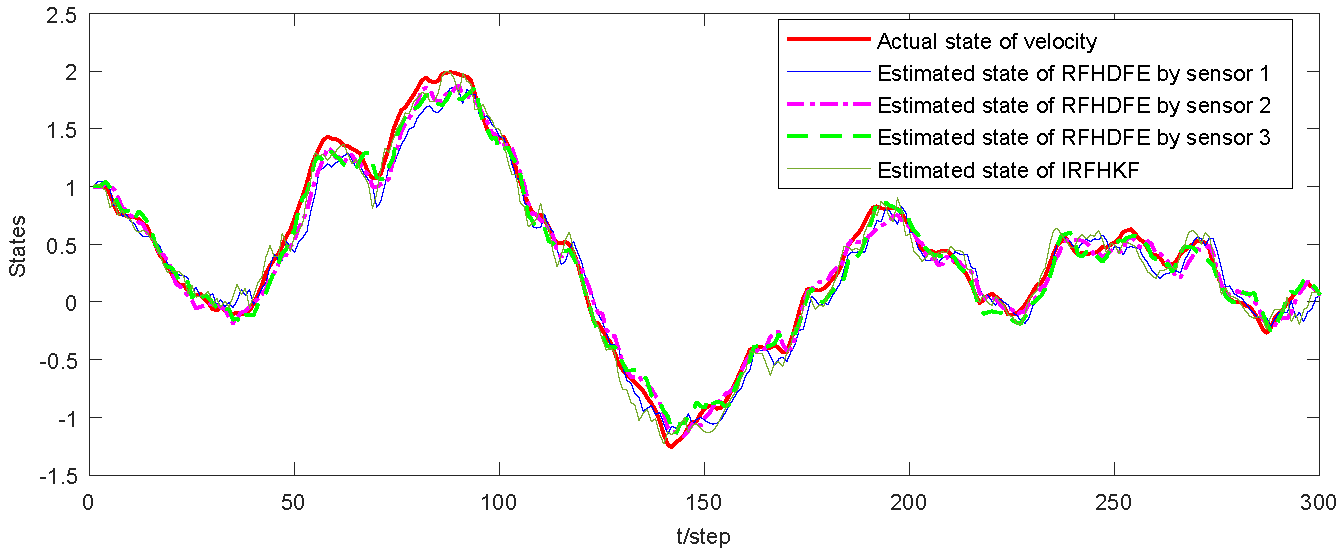}}
\\
\subfigure[State estimation of acceleration]{\includegraphics[width=0.8\linewidth]{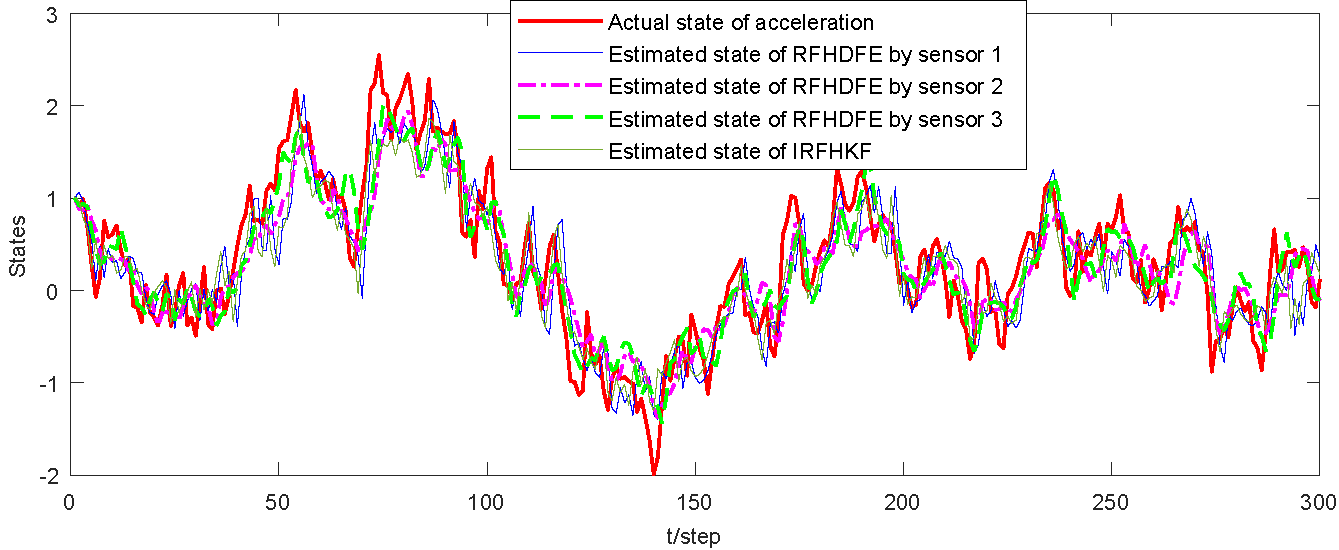}}
\\
\vspace{-4mm}
\caption{Comparison of covariance matrices between RFHDFE and IRFHKF}
}
\label{Fig 4}
\end{figure}
\end{shrinkeq}

\renewcommand{\baselinestretch}{0.48}
\begin{center}\label{Tab 1}
\vbox{\centering{\small
Table 1\quad Comparison of mean square error values}
\vskip2mm
{\footnotesize\centerline{\tabcolsep=5.2pt\begin{tabular*}{0.48\textwidth}{cccccc}
\toprule
\multicolumn{1}{c}{\multirow {2}{*}{Scheme}} & \multicolumn{1}{c}{\multirow {2}{*}{Subsystem}} &\multicolumn{3}{c}{MSEs}\\\cline{3-5}
  \multicolumn{1}{c}{}&\multicolumn{1}{c}{}&\multicolumn{1}{c}{Position}&
\multicolumn{1}{c}{Velocity} & \multicolumn{1}{c}{Acceleration} & \multicolumn{1}{c}{\multirow {-2}{*}{Time (s)}}\\
\hline
\multirow {4}{*}{RFHDFE}&
    Sensor 1  & 0.0041 & 0.0073 & 0.0335 & 0.0624 \\
  & Sensor 2  & 0.0025 & 0.0066 & 0.0289 & 0.1560 \\
  & Sensor 3	& 0.0018 & 0.0047 &	0.0245 & 0.0936 \\
  & Fusion	& 0.0004 & 0.0006 &	0.0005 & 0.3432 \\
\hline
\multirow{1}{*}{IRFHKF}  &	& 0.0124 & 0.0148 &	0.0322 & 1.8564 \\
\bottomrule
\end{tabular*}}}}
\end{center}

To illustrate the effectiveness of the system, the corresponding tracking results of the estimation for the moving target are shown in Fig.5. The simulation results come from the proposed (RFHDFE) approach with packet disorders and the IRFHKF method, which are obtained from Eqs.(\ref{eq:57}) and (\ref{eq:58}). More importantly, the actual and the estimated states with or without packet disorders are shown in Figs.5 (a)-(c), and it is a further verification that using the logic ZOH is capable of identifying and discarding packet disorders induced from random transmission delays. Meanwhile, Fig.5 (d) verifies the system stability by showing the traces for the covariance matrices of the estimation errors depending on each filter and the fusion estimation using the proposed RFHDFE method and the IRFHKF method.

\begin{shrinkeq}{-1ex}
\begin{figure}[H]
\setcounter{subfigure}{0}
{\centering
\subfigure[Estimated state for position]{\includegraphics[width=0.8\linewidth]{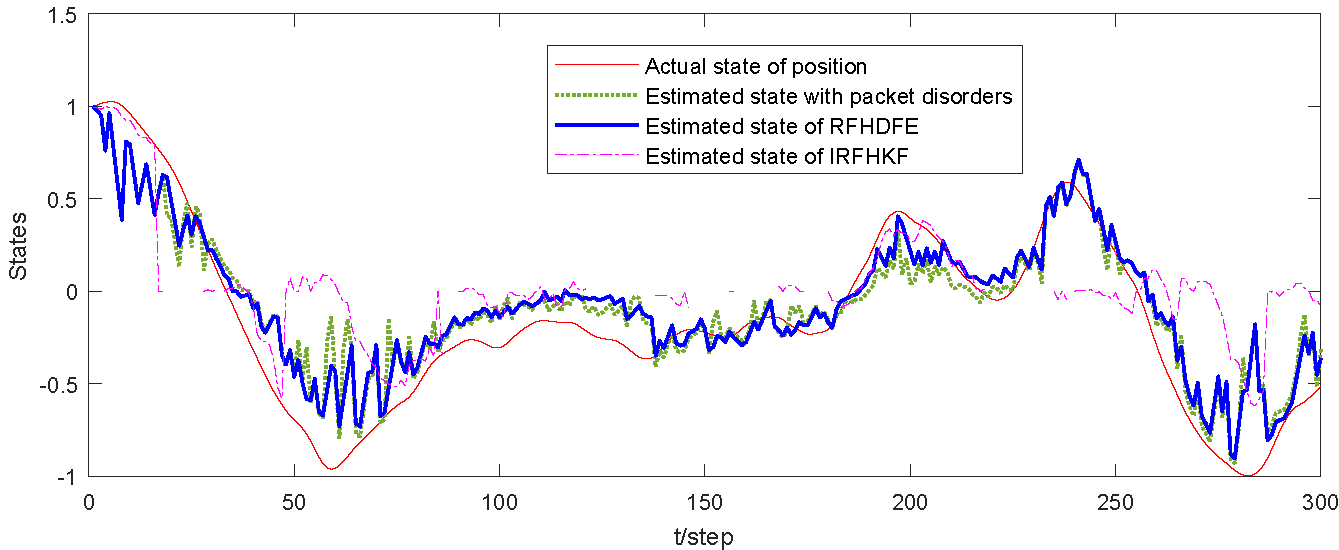}}
\\
\subfigure[Estimated state for velocity]{\includegraphics[width=0.8\linewidth]{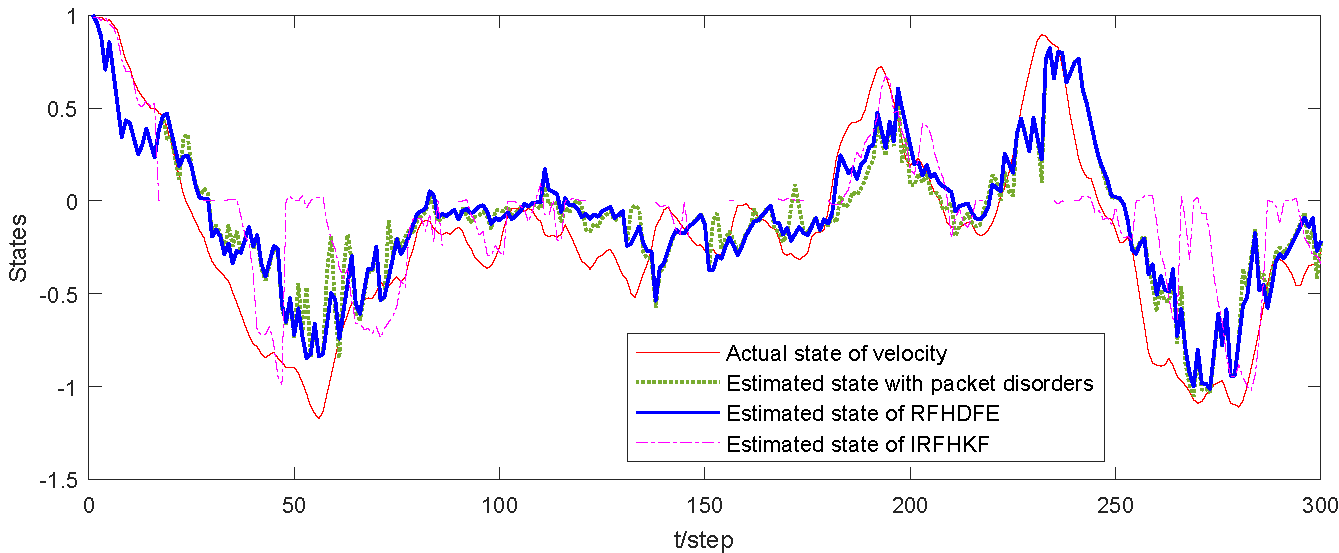}}
\\
\subfigure[Estimated state for acceleration]{\includegraphics[width=0.8\linewidth]{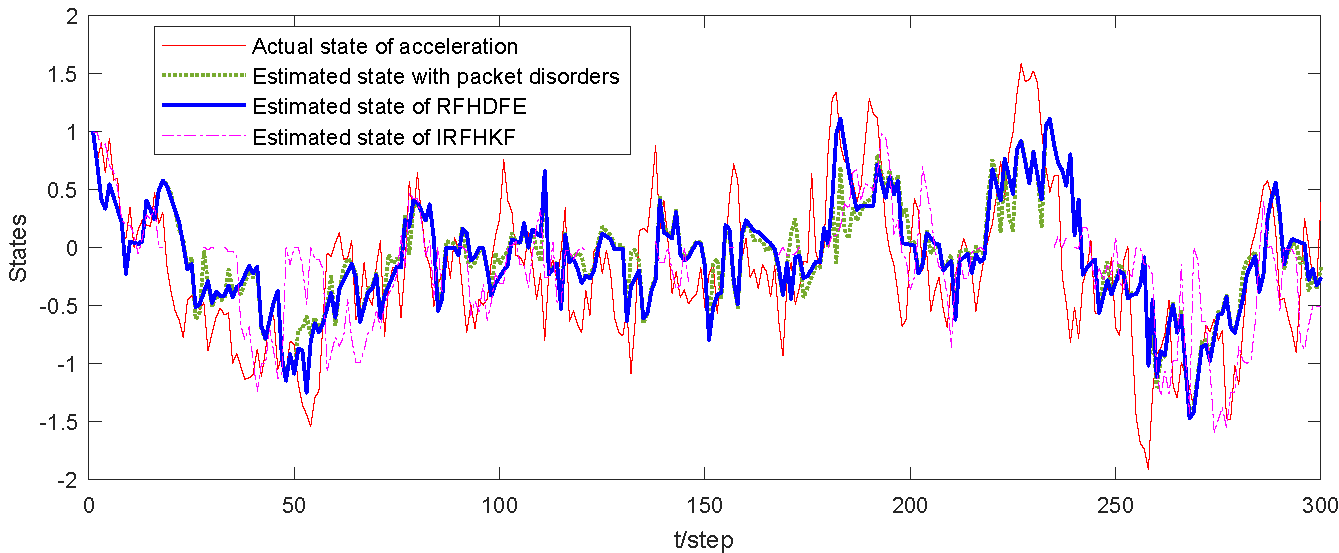}}
\\
\subfigure[Trace of error covariance matrices]{\includegraphics[width=0.8\linewidth]{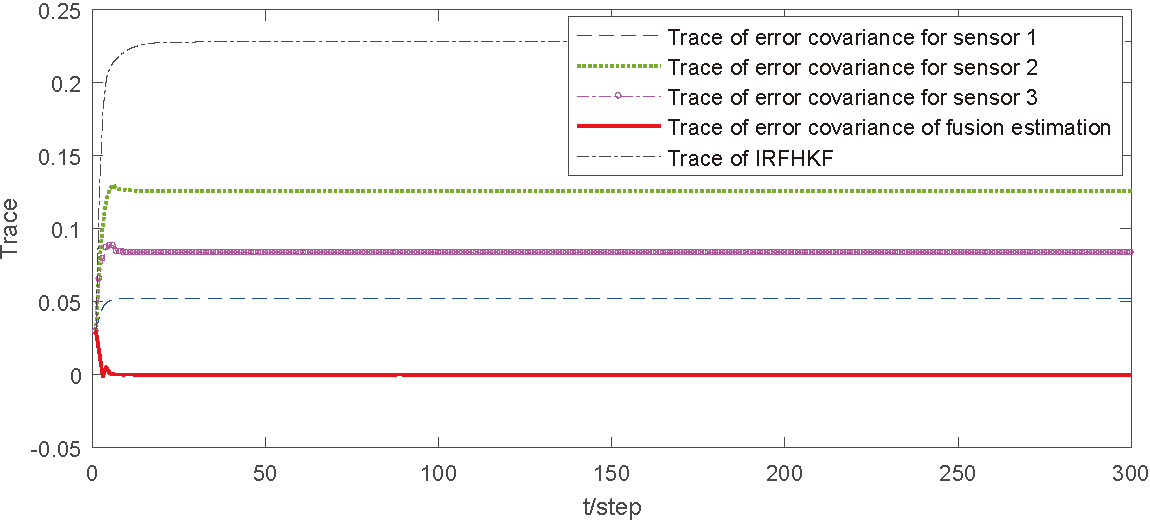}}
\\
\vspace{-4mm}
\caption{Estimated results with or without packet disorders}}
\label{Fig 5}
\end{figure}
\end{shrinkeq}

The range of the performance indicator is presented in Table E.1 of Appendix E, which means that for the proposed approach, the upper bounds for the covariance matrices of the estimation errors are less than those of the IRFHKF method. Note that the estimation performance is improved, and the estimation accuracy is higher if the upper bound for error covariance is close to the lower bound.

The distributed estimation results for the states are shown in Fig. E.1. In the simulations, the re-ordering error covariance criteria are used to compare the dynamic tracking performance between the RFHDFE and IRFHKF methods. It suggests that the proposed RFHDFE approach has better estimation accuracy than the IRFHKF method with the random transmission delays, packet dropouts and disorders simultaneously. Since the energy bound noise, unknown state-dependent noise, and the actual estimation errors are lower than the upper bound, for the actual estimated values, the proposed RFHDFE method has better performance including accurate system states and rapid convergence to the steady-state.

\section{Conclusion}
In this paper, the distributed fusion estimation problem considering a class of stochastic uncertain systems with network-induced complexity, parameter uncertainty and mixed noise disturbance have been investigated. To handle packet dropouts and packet disorders generated from random transmission delays, the data packets selection and system modelling have been implemented using the logic ZOH. For the admissible uncertainty, the problem of the local estimation has been investigated by analyzing the mean-square errors and the mix ${H_2}/{H_\infty }$ performance by means of the Schur complement lemma. To estimate the state with packet dropouts, a linear compensation method has been proposed, and the optimal state estimation has also been presented to overcome the influence of the limited communication capacity and suppress the computational burden. In addition, the upper bound of the estimation error has been designed developing the local finite horizon filtering. Subsequently, the distributed state estimation approach has been integrated with the weighted fusion criterion. Moreover, the proposed distributed fusion estimation has the advantage of higher accuracy than each local subsystem by the performance analysis of MSEs. Finally, the numerical simulations for the target tracking systems have been given to illustrate the advantages and effectiveness of the proposed RFHDFE method, in particular at the system state and the higher accuracy of the estimated value.

\section{Acknowledgements}

\vspace{-2mm}
This work was supported by National Natural Science Foundation of China (61903172, 61877065, 61872170, 52077213, 62003332, 61633016, 61673200), the Major Basic Research Project of the Natural Science Foundation of Shandong Province of China (ZR2018ZC0438), Natural Science Foundation of Guangdong (2018A030310671), and Outstanding Young Researcher Innovation Fund of Shenzhen Institute of Advanced Technology, Chinese Academy of Sciences (201822).

\appendix

\section{Description of a typical scenario.}
\setcounter{equation}{0}
\renewcommand{\theequation}{A.\arabic{equation}}
This section gives the description of a typical scenario in subsection 2.2.

Fig.2 shows a typical scenario to describe the network-induced complexity (i.e. packet dropouts and packet disorders generated from the random transmission delays).

For the $i^{th}$ subsystem, it is assumed that the upper bound of the transmission delays is not more than 5 sample periods (i.e. ${N_k} \leqslant 5$, and the delay for each step satisfies ${N_k} \leqslant k$). $\eta \left( {{t_k}} \right)$ represents the transmission delay, and the sample period is $T$, while $t \in \left\{ {kT,\;k \in \mathbb{N}} \right\}$ denotes the sampling time instant. Depending on the role of the logic ZOH, the packet disorders come from the signals before being transmitted, such as $z_{\text{3}}^i$ and $z_{\text{2}}^i$, $z_{\text{8}}^i$ and  $z_{\text{7}}^i$, as well as $z_{\text{11}}^i$, $z_{\text{10}}^i$ and $z_{\text{9}}^i$, and then $z_{\text{3}}^i$, $z_{\text{8}}^i$ and $z_{\text{11}}^i$ are held at time instant $k = 5$, $k = 9$ and $k = 12$, respectively.

\section{Augmented state-space model.}
\setcounter{equation}{0}
\renewcommand{\theequation}{B.\arabic{equation}}
This section gives the definitions and derivations of the augmented state-space model shown in subsection 3.1.

An augmented state-space model combining the systems shown in Eqs.(\ref{eq:1}) and (\ref{eq:17})-(\ref{eq:19}) is represented as follows:
\begin{sequation}\label{eq:B.1}
\Psi _t^i = \left( {A_{t1}^i + H_{t1}^i{F_t}E_{t1}^i} \right)\tilde \Psi _t^i + B_{t1}^i{w_t} + G_{t1}^i{v_t}\;,
\end{sequation}
and
\begin{sequation}\label{eq:B.2}
\begin{gathered}
  \tilde \Psi _{t + 1}^i = \left( {A_{t2}^i + H_{t2}^i{F_t}E_{t2}^i + \sum\limits_{\vartheta  = 1}^\hbar  {{A_{\vartheta ,t2}}{\varpi _{\vartheta ,t}}} } \right)\tilde \Psi _t^i \hfill \\
  \;\; + B_{t2}^i{w_t} + G_{t2}^i{v_t}\;, \hfill \\
\end{gathered}
\end{sequation}
with
\[\begin{gathered}
  A_{t1}^i = \left[ \begin{gathered}
  I - K_t^iC_t^i\;\;\;\;\;\;\;\;K_t^i\left( {\hat C_t^i - C_t^i} \right) \hfill \\
  K_t^iC_t^i\;\;\;\;\;\;\;\;I + K_t^i\left( {C_t^i - \hat C_t^i} \right) \hfill \\
\end{gathered}  \right]\;,\; \hfill \\
  H_{t1}^i = \left[ \begin{gathered}
   - K_t^i\mathcal{H}_t^i \hfill \\
  K_t^i\mathcal{H}_t^i \hfill \\
\end{gathered}  \right]\;,\;E_{t1}^i = E_{t2}^i = \left[ {E_t^i\;\;\;E_t^i} \right],\; \hfill \\
  B_{t1}^i = \left[ \begin{gathered}
   - K_t^iB_t^i \hfill \\
  K_t^iB_t^i \hfill \\
\end{gathered}  \right],\;G_{t1}^i = \left[ \begin{gathered}
   - K_t^iG_t^i \hfill \\
  K_t^iG_t^i \hfill \\
\end{gathered}  \right], \hfill \\
  A_{t2}^i = \left[ \begin{gathered}
  {A_t} - L_t^iC_t^i\;\;\;\;{A_t} - \hat A_t^i + L_t^i\left( {\hat C_t^i - C_t^i} \right) \hfill \\
  L_t^iC_t^i\;\;\;\;\;\;\;\;\;\;\;\;\hat A_t^i + L_t^i\left( {C_t^i - \hat C_t^i} \right) \hfill \\
\end{gathered}  \right]\;, \hfill \\
\end{gathered} \]

\begin{sequation}\label{eq:B.3}
\begin{gathered}
  H_{t2}^i = \left[ \begin{gathered}
  {\mathcal{F}_t} - L_t^i\mathcal{H}_t^i \hfill \\
  L_t^i\mathcal{H}_t^i \hfill \\
\end{gathered}  \right],\;{A_{\vartheta ,t2}} = {A_\vartheta }\left[ \begin{gathered}
  I\;\;\;\;I \hfill \\
  0\;\;\;\;0 \hfill \\
\end{gathered}  \right],\; \hfill \\
  B_{t2}^i = \left[ \begin{gathered}
  {B_t} - L_t^iB_t^i \hfill \\
  L_t^iB_t^i \hfill \\
\end{gathered}  \right],\;G_{t2}^i = \left[ \begin{gathered}
  {G_t} - L_t^iG_t^i \hfill \\
  L_t^iG_t^i \hfill \\
\end{gathered}  \right]\;. \hfill \\
\end{gathered}
\end{sequation}

Based on the augmented system from Eqs.(\ref{eq:B.1}) and (\ref{eq:B.2}), we set the covariance matrices to be $\tilde \Sigma _t^i = E\left( {\tilde \Psi _t^i{{\left( {\tilde \Psi _t^i} \right)}^T}} \right)$ and $\tilde \Theta _t^i = E\left( {\Psi _t^i{{\left( {\Psi _t^i} \right)}^T}} \right)$ under Eqs.(\ref{eq:B.1})-(\ref{eq:B.3}). Then, the Riccati-like equations for the covariance matrices of the estimation errors are derived as follows:

\begin{shrinkeq}{-1.5ex}
\begin{sequation}\label{eq:B.4}
\begin{gathered}
  \tilde \Theta _t^i = \left( {A_{t1}^i + H_{t1}^i{F_t}E_{t1}^i} \right)\tilde \Sigma _t^i{\left( {A_{t1}^i + H_{t1}^i{F_t}E_{t1}^i} \right)^T} \hfill \\
  \;\; + B_{t1}^iE\left( {{w_t}w_t^T} \right){\left( {B_{t1}^i} \right)^T} + G_{t1}^i{R_t}{\left( {G_{t1}^i} \right)^T}\;, \hfill \\
\end{gathered}
\end{sequation}
and
\begin{sequation}\label{eq:B.5}
\begin{gathered}
  \tilde \Sigma _{t + 1}^i = \left( {A_{t2}^i + H_{t2}^i{F_t}E_{t2}^i} \right)\tilde \Sigma _t^i{\left( {A_{t2}^i + H_{t2}^i{F_t}E_{t2}^i} \right)^T} \hfill \\
  \;\; + \sum\limits_{\vartheta  = 1}^\hbar  {{A_{\vartheta ,t2}}\tilde \Sigma _t^iA_{\vartheta ,t2}^T{\theta _{\vartheta ,k}}}  + B_{t2}^iE\left( {{w_t}w_t^T} \right){\left( {B_{t2}^i} \right)^T} \hfill \\
  \;\; + G_{t2}^i{R_t}{\left( {G_{t2}^i} \right)^T}\;. \hfill \\
\end{gathered}
\end{sequation}
\end{shrinkeq}

\section{Proof of Theorem 1 and derivation of the convex optimization problem.}
\setcounter{equation}{0}
\setcounter{remark}{0}
\renewcommand{\theequation}{C.\arabic{equation}}
\renewcommand{\theremark}{C.\arabic{remark}}
This section gives the proof of Theorem 1 shown in Section 3.2.

For the $i^{th}$ subsystem in Eq.(\ref{eq:20}) with ${v_k} = 0$, it is derived from Lemma 1 that the requirements (i.e. Conditions 1 and 2) are held, if and only if there exists a matrix ${{\rm X}_i} > 0$ such that
\begin{equation}\label{eq:C.1}
\begin{gathered}
  {\left( {A_{t3}^i + \Delta A_{t3}^i} \right)^T}{{\rm X}_i}\left( {A_{t3}^i + \Delta A_{t3}^i} \right) + {\left( {D_{t3}^i} \right)^T}D_{t3}^i - {{\rm X}_i} \hfill \\
  \;\; + \sum\limits_{\vartheta  = 1}^\hbar  {{\varpi _{\vartheta ,k}}A_{\vartheta ,t3}^T{{\rm X}_i}{A_{\vartheta ,t3}}}  \hfill \\
  \;\; + {\left( {A_{t3}^i + \Delta A_{t3}^i} \right)^T}{{\rm X}_i}B_{t3}^i{\left( {\gamma _i^2 - {{\left( {B_{t3}^i} \right)}^T}{{\rm X}_i}B_{t3}^i} \right)^{ - 1}} \hfill \\
  \;\; \times {\left( {B_{t3}^i} \right)^T}{{\rm X}_i}\left( {A_{t3}^i + \Delta A_{t3}^i} \right)\; < 0\;. \hfill \\
\end{gathered}
\end{equation}

Thus, the inequality (\ref{eq:27}) is obtained from Eq.(\ref{eq:C.1}) using the Schur complement lemma \cite{Bib41}. On the other hand, when the system following Eq.(\ref{eq:20}) is mean-square stable, the ${H_2}$ performance ${J^i}$ can be expressed as follows \cite{Bib42}:
\begin{equation}\label{eq:C.2}
{J^i} = \mathop {\lim }\limits_{t \to \infty } E\left( {{{\left( {\tilde e_{k + 1}^i} \right)}^T}\tilde e_{k + 1}^i} \right) = Tr\left( {{R_k}{{\left( {G_{t3}^i} \right)}^T}{{\bar {\rm X}}_i}G_{t3}^i} \right)
\end{equation}

\noindent
where ${\bar {\rm X}_i}$ is the solution of the following Lyapunov equation:
\begin{equation}\label{eq:C.3}
\begin{gathered}
  {\left( {A_{t3}^i + \Delta A_{t3}^i} \right)^T}{{\bar {\rm X}}_i}\left( {A_{t3}^i + \Delta A_{t3}^i} \right) - {{\bar {\rm X}}_i} \hfill \\
  \;\; + {\left( {D_{t3}^i} \right)^T}D_{t3}^i + \sum\limits_{\vartheta  = 1}^\hbar  {{\varpi _{\vartheta ,k}}A_{\vartheta ,t3}^T{{\bar {\rm X}}_i}{A_{\vartheta ,t3}}}  = 0\;. \hfill \\
\end{gathered}
\end{equation}

Meanwhile, it is known from Eq.(\ref{eq:27}) that
\begin{equation}\label{eq:C.4}
\begin{gathered}
  {\left( {A_{t3}^i + \Delta A_{t3}^i} \right)^T}{{\rm X}_i}\left( {A_{t3}^i + \Delta A_{t3}^i} \right) - {{\rm X}_i} \hfill \\
  \;\; + {\left( {D_{t3}^i} \right)^T}D_{t3}^i + \sum\limits_{\vartheta  = 1}^\hbar  {{\varpi _{\vartheta ,k}}A_{\vartheta ,t3}^T{{\rm X}_i}{A_{\vartheta ,t3}}}  < 0\;. \hfill \\
\end{gathered}
\end{equation}

Then, it is concluded that ${\bar {\rm X}_i} \leqslant {{\rm X}_i}$. In this case, the upper bound of ${H_2}$ performance ${J^i}$ can be treated as $Tr\left( {{R_k}{{\left( {G_{t3}^i} \right)}^T}{{\rm X}_i}G_{t3}^i} \right)$, where ${{\rm X}_i}$ is the solution to the matrix inequality based on Eq.(\ref{eq:27}).

Complete the proof of Theorem 1.

Based on Lemma 2, we define $W_i^T \triangleq \left[ \begin{gathered}
  {W_{i,1}}\;\;\;\;\;0 \hfill \\
  {W_{i,3}}\;\;{W_{i,4}} \hfill \\
\end{gathered}  \right]$, and the inequality from Eq.(\ref{eq:27}) will hold. If there exists a matrix ${W_i}$, meanwhile, the following inequality
\begin{sequation}\label{eq:C.5}
\left[ \begin{gathered}
  {{\rm X}_{{W_i}}}\;\;\;\;{A_{{W_i}}} + {F_{{W_i}}}{F_k}{E_{{W_i}}}\;\;\;\;\;{B_{{W_i}}} \hfill \\
  *\;\;\;{\left( {D_{t3}^i} \right)^T}D_{t3}^i + {\Upsilon _{t3}} - {{\rm X}_i}\;\;\;0 \hfill \\
  *\;\;\;\;\;\;\;\;\;\;\;\;\;\;*\;\;\;\;\;\;\;\;\;\;\;\;\;\; - \gamma _i^2I \hfill \\
\end{gathered}  \right] < 0
\end{sequation}

\noindent
holds, and the form of the inequality is similar to Eq.(\ref{eq:27}). Note that parameters are defined as follows:

${\Upsilon _{t3}}$ is defined in Eq.(\ref{eq:27}), and
\begin{sequation}\label{eq:C.6}
\begin{gathered}
  {{\rm X}_{{W_i}}} = {{\rm X}_i} - {W_i} - W_i^T\;, \hfill \\
  {A_{{W_i}}} = \left[ \begin{gathered}
  {W_{i,1}}{A_k}\;\;\;\;\;\;\;\;\;\;\;\;\;\;\;\;\;\;\;\;\;0 \hfill \\
  {A_{{W_{i,3}}}}\;\;\;\;{W_{i,4}}\left( {\hat A_k^i - L_k^i\hat C_k^i} \right) \hfill \\
\end{gathered}  \right]\;, \hfill \\
  {B_{{W_i}}} = \left[ \begin{gathered}
  {W_{i,1}}{B_k} \hfill \\
  \left( {{W_{i,3}} + {W_{i,4}}} \right){B_k} - {W_{i,4}}L_k^iB_k^i \hfill \\
\end{gathered}  \right]\;, \hfill \\
  {F_{{W_i}}} = \left[ \begin{gathered}
  {W_{i,1}}{\mathcal{F}_k} \hfill \\
  {W_{i,3}}{\mathcal{F}_k} + {W_{i,4}}\left( {{\mathcal{F}_k} - L_k^i\mathcal{H}_k^i} \right) \hfill \\
\end{gathered}  \right]\;,\;{E_{{W_i}}} = \left[ {{E_k}\;\;0} \right]\;, \hfill \\
  {A_{{W_{i,3}}}} = \left( {{W_{i,3}} + {W_{i,3}}} \right){A_k} + {W_{i,4}}\left( {L_k^i\left( {\hat C_k^i - C_k^i} \right) - \hat A_k^i} \right)\;. \hfill \\
\end{gathered}
\end{sequation}

From Theorem 1, the ${H_2}$ performance in Eq.(\ref{eq:20}) satisfies ${J^i} \leqslant Tr\left( {{R_k}{{\left( {G_{t3}^i} \right)}^T}{{\rm X}_i}G_{t3}^i} \right)$. Then, the upper bound of a symmetric matrix ${\rho _{i,0}}$ is introduced to conform to ${R_k}{\left( {G_{t3}^i} \right)^T}{{\rm X}_i}G_{t3}^i \leqslant {\rho _{i,0}}$. Thus, the inequality ${\left( {G_{t3}^i} \right)^T}{{\rm X}_i}G_{t3}^i \leqslant R_k^{ - 1}{\rho _{i,0}}$ is satisfied. Then, we define ${\rho _i} \triangleq R_k^{ - 1}{\rho _{i,0}}$, based on the Schur complement lemma, the inequality ${\left( {G_{t3}^i} \right)^T}{{\rm X}_i}G_{t3}^i \leqslant {\rho _i}$ is equivalent to
\begin{sequation}\label{eq:C.7}
\left[ \begin{gathered}
   - {{\rm X}_i}\;\;\;{{\rm X}_i}G_{t3}^i \hfill \\
   * \;\;\;\;\;\;\;\;\; - {\rho _i} \hfill \\
\end{gathered}  \right] < 0\;.
\end{sequation}

\noindent
It means that the inequality given in Eq.(\ref{eq:C.7}) is held, if there exists a matrix ${W_i}$, which meets the following linear matrix inequalities (LMIs) condition:
\begin{sequation}\label{eq:C.8}
\left[ \begin{gathered}
  {{\rm X}_{{W_i}}}\;\;\;\;{G_{{W_i}}} \hfill \\
  *\;\;\;\;\;\;\; - {\rho _i} \hfill \\
\end{gathered}  \right] < 0\;,
\end{sequation}

\noindent
where ${{\rm X}_{{W_i}}}$ is given in Eq.(\ref{eq:C.6}), while ${G_{{W_i}}} = \left[ \begin{gathered}
  {W_{i,1}}{G_k} \hfill \\
  \left( {{W_{i,3}} + {W_{i,4}}} \right){G_k} - {W_{i,4}}L_k^iG_k^i \hfill \\
\end{gathered}  \right]$.

\begin{remark}\label{Remark C.1}
Note that for the matrix ${W_i}$, if there is no structural constraint, the inequality in Eq.(\ref{eq:C.5}) is equivalent to Eq.(\ref{eq:27}), and the inequality in Eq.(\ref{eq:C.8}) will be equivalent to Eq.(\ref{eq:C.7}). However, the nonlinear terms in Eqs.(\ref{eq:27}) and (\ref{eq:C.7}) are unable to be eliminated in this case. For this reason, an equivalent LMI will be given, which is used to represent the inequality in Eq.(\ref{eq:C.5}), and then the local estimation parameters will be obtained by solving a convex optimization problem.
\end{remark}

\section{Proof of Theorem 3.}
\setcounter{equation}{0}
\renewcommand{\theequation}{D.\arabic{equation}}
This section gives the proof of Theorem 3 shown in Section 4.1.

The solutions of $\bar \Theta _t^i$ and $\bar \Sigma _{t + 1}^i$ are derived from $\bar \Sigma _t^i$ in Eqs.(\ref{eq:35}) and (\ref{eq:36}), so that the upper bound of $\Sigma _t^i$ under Eq.(\ref{eq:34}) can be represented as follows \cite{Bib19,Bib22}:
\begin{equation}\label{eq:D.1}
\Sigma _t^i = \left[ \begin{gathered}
  {\Sigma _{11,\;t}}\;\;\;{\Sigma _{12,\;t}} \hfill \\
  {\Sigma _{21,\;t}}\;\;\;{\Sigma _{22,\;t}} \hfill \\
\end{gathered}  \right] = \left[ \begin{gathered}
  \bar \Sigma _t^i\;\;\;\;\;\;\;\;0 \hfill \\
  0\;\;\;\;{P_t} - \bar \Sigma _t^i \hfill \\
\end{gathered}  \right]\;.
\end{equation}

To estimate the filter parameters $\hat C_t^i$, $K_t^i$, $\hat A_t^i$ and $L_t^i$, considering the given recursive equations for $\bar \Sigma _{t + 1}^i$ and ${P_{t + 1}}$ in Eqs.(\ref{eq:43}) and (\ref{eq:44}), the approach of optimizing measurement and filtering error covariance matrices is developed based on the local estimators $\hat x_{t|t}^i$ and $\hat x_{t + 1|t}^i$ in Eqs.(\ref{eq:17}) and (\ref{eq:18}), respectively. Therefore, the derivation process is shown as follows:

\textbf{Step 1:} Solve the filter parameter $\hat C_t^i$.

Due to setting $t = k - {\tau ^i}\left( {{k_1}} \right)$, the measurement error $\tilde y_k^i$ is defined as:
\begin{equation}\label{eq:D.2}
\begin{gathered}
  \tilde y_k^i = y_k^i - \hat y_k^i \hfill \\
   = \left( {A_{t4}^i + H_{t4}^i{F_t}E_{t4}^i} \right)\tilde \Psi _t^i + B_t^i{w_t} + G_t^i{v_t}\;, \hfill \\
\end{gathered}
\end{equation}
where
\begin{equation}\label{eq:D.3}
A_{t4}^i = \left[ {C_t^i\;\;\;C_t^i - \hat C_t^i} \right]\;,\;H_{t4}^i = \mathcal{H}_t^i,\;E_{t4}^i = \left[ {E_t^i\;\;\;E_t^i} \right]\;.
\end{equation}

Next, obtain the upper bound for the covariance of the measurement error from Eq.(\ref{eq:D.1}), and Lemmas 3 and 4:
\begin{equation}\label{eq:D.4}
\begin{gathered}
  E\left( {\tilde y_k^i{{\left( {\tilde y_k^i} \right)}^T}} \right) \hfill \\
   \leqslant A_{t4}^i\Sigma _t^i{\left( {A_{t4}^i} \right)^T} + \alpha _t^{ - 1}H_{t4}^i{\left( {H_{t4}^i} \right)^T} \hfill \\
  \;\; + B_t^iE\left( {{w_t}w_t^T} \right) + G_t^i{R_t}{\left( {G_t^i} \right)^T} \hfill \\
  \;\; + A_{t4}^i\Sigma _t^i{\left( {E_{t4}^i} \right)^T}{\left( {\alpha _t^{ - 1}I - E_{t4}^i\Sigma _t^i{{\left( {E_{t4}^i} \right)}^T}} \right)^{ - 1}}E_{t4}^i\Sigma _t^i{\left( {A_{t4}^i} \right)^T} \hfill \\
   = \bar \Pi _t^i \hfill \\
   = C_t^i\bar \Sigma _t^i{\left( {C_t^i} \right)^T} + \left( {C_t^i - \hat C_t^i} \right)\left( {{P_t} - \bar \Sigma _t^i} \right){\left( {C_t^i - \hat C_t^i} \right)^T} \hfill \\
  \;\; + \alpha _t^{ - 1}\mathcal{H}_t^i{\left( {\mathcal{H}_t^i} \right)^T} + B_t^iE\left( {{w_t}w_t^T} \right) + G_t^i{R_t}{\left( {G_t^i} \right)^T} \hfill \\
  \;\; + \left( {C_t^i\bar \Sigma _t^i{{\left( {E_t^i} \right)}^T} + \left( {C_t^i - \hat C_t^i} \right)\left( {{P_t} - \bar \Sigma _t^i} \right){{\left( {E_t^i} \right)}^T}} \right) \hfill \\
  \;\; \times {\left( {\alpha _t^{ - 1}I - E_t^i{P_t}{{\left( {E_t^i} \right)}^T}} \right)^{ - 1}} \hfill \\
  \;\; \times {\left( {C_t^i\bar \Sigma _t^i{{\left( {E_t^i} \right)}^T} + \left( {C_t^i - \hat C_t^i} \right)\left( {{P_t} - \bar \Sigma _t^i} \right){{\left( {E_t^i} \right)}^T}} \right)^T}\;. \hfill \\
\end{gathered}
\end{equation}

\noindent
Therefore, we use the first order derivative $\displaystyle\frac{{\partial \bar \Pi _t^i}}{{\partial \hat C_t^i}} = 0$ to obtain $\hat C_t^i$:
\begin{equation}\label{eq:D.5}
\hat C_t^i = C_t^i\left( {I + \bar \Sigma _t^i{{\left( {E_t^i} \right)}^T}{{\left( {M_t^i} \right)}^{ - 1}}E_t^i} \right)\;,
\end{equation}
where $M_t^i = \alpha _t^{ - 1}I - E_t^i\bar \Sigma _t^i{\left( {E_t^i} \right)^T} > 0$.

Finally, similar to the derivation of $\hat C_t^i$, the other filter parameters such as $K_t^i$, $\hat A_t^i$ and $L_t^i$ are generated.

\textbf{Step 2:} Derive the error covariance matrices $\bar \Theta _t^i$, $\bar \Sigma _{t + 1}^i$ and ${P_{t + 1}}$, respectively.

Firstly, Theorem 2 defines the solutions of $\Theta _t^i$ and $\Sigma _{t + 1}^i$ in Eqs.(\ref{eq:32}) and (\ref{eq:33}). Subsequently, the upper bounds for the covariance matrices of the estimation errors $\bar \Theta _t^i$ and $\bar \Sigma _{t + 1}^i$ from Eqs.(\ref{eq:35}) and (\ref{eq:36}) are derived as follows:
\begin{equation}\label{eq:D.6}
\begin{gathered}
  \bar \Theta _t^i = \left( {I - K_t^iC_t^i} \right)\bar \Sigma _t^i\;{\left( {I - K_t^iC_t^i} \right)^T} \hfill \\
  \;\; + K_t^i\left( {\hat C_t^i - C_t^i} \right)\left( {{P_t} - \bar \Sigma _t^i} \right){\left( {K_t^i\left( {\hat C_t^i - C_t^i} \right)} \right)^T} \hfill \\
  \;\; + \;\alpha _t^{ - 1}K_t^i\mathcal{H}_t^i{\left( {K_t^i\mathcal{H}_t^i} \right)^T} + K_t^iG_t^i{R_t}{\left( {K_t^iG_t^i} \right)^T} \hfill \\
  \;\; + K_t^iB_t^iE\left( {{w_t}w_t^T} \right){\left( {K_t^iB_t^i} \right)^T} \hfill \\
  \;\; + \left( {\bar \Sigma _t^i + K_t^i\left( {\hat C_t^i\left( {{P_t} - \bar \Sigma _t^i} \right) - C_t^i{P_t}} \right)} \right) \hfill \\
  \;\; \times {\left( {E_t^i} \right)^T}{\left( {\tilde M_t^i} \right)^{ - 1}}E_t^i \hfill \\
  \;\; \times {\left( {\bar \Sigma _t^i + K_t^i\left( {\hat C_t^i\left( {{P_t} - \bar \Sigma _t^i} \right) - C_t^i{P_t}} \right)} \right)^T}, \hfill \\
\end{gathered}
\end{equation}
and
\begin{equation}\label{eq:D.7}
\begin{gathered}
  \bar \Sigma _{t + 1}^i = \left( {{A_t} - L_t^iC_t^i} \right)\bar \Sigma _t^i{\left( {{A_t} - L_t^iC_t^i} \right)^T} \hfill \\
  \;\; + \left( {{A_t} - \hat A_t^i + L_t^i\left( {\hat C_t^i - C_t^i} \right)} \right) \hfill \\
  \;\; \times \left( {{P_t} - \bar \Sigma _t^i} \right){\left( {{A_t} - \hat A_t^i + L_t^i\left( {\hat C_t^i - C_t^i} \right)} \right)^T} \hfill \\
  \;\; + \alpha _t^{ - 1}\left( {{\mathcal{F}_t} - L_t^i\mathcal{H}_t^i} \right){\left( {{\mathcal{F}_t} - L_t^i\mathcal{H}_t^i} \right)^T} \hfill \\
  \;\; + \sum\limits_{\vartheta  = 1}^\hbar  {{A_\vartheta }} {P_t}\left( {A_\vartheta ^T} \right){\theta _{\vartheta ,t}} \hfill \\
  \;\; + \left( {\left( {{A_t} - L_t^iC_t^i} \right)\bar \Sigma _t^i} \right. \hfill \\
  \;\;\;\;\;\left. { + \left( {{A_t} - \hat A_t^i + L_t^i\left( {\hat C_t^i - C_t^i} \right)} \right)\left( {{P_t} - \bar \Sigma _t^i} \right)} \right) \hfill \\
  \;\; \times {\left( {E_t^i} \right)^T}{\left( {\tilde M_t^i} \right)^{ - 1}}E_t^i \hfill \\
  \;\; \times \left( {\left( {{A_t} - L_t^iC_t^i} \right)\bar \Sigma _t^i} \right. \hfill \\
  \;\;\;\;\;{\left. { + \left( {{A_t} - \hat A_t^i + L_t^i\left( {\hat C_t^i - C_t^i} \right)} \right)\left( {{P_t} - \bar \Sigma _t^i} \right)} \right)^T} \hfill \\
  \;\; + \left( {{B_t} - L_t^iB_t^i} \right)E\left( {{w_t}w_t^T} \right){\left( {{B_t} - L_t^iB_t^i} \right)^T} \hfill \\
  \;\; + \left( {{G_t} - L_t^iG_t^i} \right){R_t}{\left( {{G_t} - L_t^iG_t^i} \right)^T}\;. \hfill \\
\end{gathered}
\end{equation}
in which $\tilde M_t^i = \alpha _t^{ - 1}I - E_t^i{P_t}{\left( {E_t^i} \right)^T}$.

According to the above derivation, introducing filter parameters $\hat C_t^i$, $K_t^i$, $\hat A_t^i$ and $L_t^i$ derived from Eqs.(\ref{eq:45})-(\ref{eq:48}), and they are substituted into the upper bounds $\bar \Theta _t^i$ and $\bar \Sigma _{t + 1}^i$ from Eqs.(\ref{eq:D.6}) and (\ref{eq:D.7}), respectively. Therefore, the error covariance matrices $\bar \Theta _t^i$ and $\bar \Sigma _{t + 1}^i$ are rewritten as:
\begin{equation}\label{eq:D.8}
\bar \Theta _t^i = \bar \Sigma _t^i + \bar \Sigma _t^i{\left( {E_t^i} \right)^T}{\left( {\tilde M_t^i} \right)^{ - 1}}E_t^i\bar \Sigma _t^i - \Lambda _t^i{\left( {\Xi _t^i} \right)^{ - 1}}{\left( {\Lambda _t^i} \right)^T}\;,
\end{equation}
and
\begin{equation}\label{eq:D.9}
\bar \Sigma _{t + 1}^i = \tilde \Delta _t^i - L_t^i\tilde \nabla _t^i\;,
\end{equation}
where
\[\Lambda _t^i = \left( {I + \bar \Sigma _t^i{{\left( {E_t^i} \right)}^T}{{\left( {M_t^i} \right)}^{ - 1}}E_t^i} \right)\bar \Sigma _t^i{\left( {C_t^i} \right)^T}\;,\]

\[\begin{gathered}
  \Xi _t^i = C_t^i\bar \Sigma _t^i\;\left( {I + {{\left( {E_t^i} \right)}^T}{{\left( {M_t^i} \right)}^{ - 1}}E_t^i\bar \Sigma _t^i} \right){\left( {C_t^i} \right)^T} \hfill \\
  \;\; + \alpha _t^{ - 1}\mathcal{H}_t^i{\left( {\mathcal{H}_t^i} \right)^T} + B_t^iE\left( {{w_t}w_t^T} \right){\left( {B_t^i} \right)^T} + G_t^i{R_t}{\left( {G_t^i} \right)^T}\;, \hfill \\
\end{gathered} \]

\[\begin{gathered}
  \tilde \Delta _t^i = {A_t}\left( {I + \bar \Sigma _t^i{{\left( {E_t^i} \right)}^T}{{\left( {M_t^i} \right)}^{ - 1}}E_t^i} \right)\bar \Sigma _t^iA_t^T + \alpha _t^{ - 1}{\mathcal{F}_t}\mathcal{F}_t^T \hfill \\
  \;\; + {B_t}E\left( {{w_t}w_t^T} \right)B_t^T + {G_t}{R_t}G_t^T + \sum\limits_{\vartheta  = 1}^\hbar  {{A_\vartheta }} {P_t}\left( {A_\vartheta ^T} \right){\theta _{\vartheta ,t}}\;, \hfill \\
\end{gathered} \]
and \[\begin{gathered}
  \tilde \nabla _t^i = C_t^i\left( {I + \bar \Sigma _t^i{{\left( {E_t^i} \right)}^T}{{\left( {M_t^i} \right)}^{ - 1}}E_t^i} \right)\bar \Sigma _t^iA_t^T \hfill \\
  \;\; + \alpha _t^{ - 1}\mathcal{H}_t^i\mathcal{F}_t^T + B_t^iE\left( {{w_t}w_t^T} \right)B_t^T + G_t^i{R_t}G_t^T\;. \hfill \\
\end{gathered} \]

Subsequently, the covariance for the state with time-varying parametric uncertainty is denoted as follows:
\begin{equation}\label{eq:D.10}
\begin{gathered}
  {{\tilde P}_{t + 1}} = E\left( {{x_{t + 1}}x_{t + 1}^T} \right) \hfill \\
   = \left( {{A_t} + {\mathcal{F}_t}{F_t}{E_t}} \right){{\tilde P}_t}{\left( {{A_t} + {\mathcal{F}_t}{F_t}{E_t}} \right)^T} \hfill \\
  \;\; + {\theta _{\vartheta ,k}}\sum\limits_{\vartheta  = 1}^\hbar  {{A_\vartheta }{{\tilde P}_t}A_\vartheta ^T}  \hfill \\
  \;\; + {B_t}E\left( {{w_t}w_t^T} \right)B_t^T + {G_t}{R_t}G_t^T\;. \hfill \\
\end{gathered}
\end{equation}

Afterwards, the upper bound for the covariance matrix of the state is obtained from
\begin{equation}\label{eq:D.11}
\begin{gathered}
  {{\tilde P}_{t + 1}} \leqslant {A_t}{\left( {P_t^{ - 1} - {\alpha _t}E_t^T{E_t}} \right)^{ - 1}}A_t^T \hfill \\
  \;\; + \alpha _t^{ - 1}{\mathcal{F}_t}\mathcal{F}_t^T + {\theta _{\vartheta ,k}}\sum\limits_{\vartheta  = 1}^\hbar  {{A_\vartheta }{P_t}A_\vartheta ^T}  \hfill \\
  \;\; + {B_t}E\left( {{w_t}w_t^T} \right)B_t^T + {G_t}{R_t}G_t^T \hfill \\
   = {P_{t + 1}} \hfill \\
\end{gathered}
\end{equation}
with the initial value ${P_0} = {x_0}x_0^T + {P_0}$, which is similarly calculated using the method reported in \cite{Bib12}.

\section{Shown and described for the figures and tables of the numerical simulation.}
\setcounter{equation}{0}
\renewcommand{\theequation}{I.\arabic{equation}}

This section shows and describes the distributed estimation results using the logic ZOH presented in Section 5.

The range of the performance indicator is presented in Table E.1.

\begin{center}\label{Tab E.1}
\vbox{\centering{\small
Table E.1\quad Comparison of the upper and lower bounds for error covariance}
\vskip2mm
\renewcommand{\baselinestretch}{0.5}
{\footnotesize\centerline{\tabcolsep=4pt\begin{tabular*}{0.5\textwidth}{ccccc}
\toprule
Method & Position & Velocity & Acceleration & Trace\\
\hline
    RFHDFE 1 & 0.0096-0.0150 & 0.0016-0.0100 & 0.0100-0.0372 & 0.0300-0.0533\\
    RFHDFE 2 & 0.0056-0.0100 & 0.0033-0.0100 & 0.0100-0.1258 & 0.0291-0.1313\\
    RFHDFE 3 & 0.0084-0.0109 & 0.0004-0.0100 & 0.0100-0.0823 & 0.0279-0.0883\\
    IRFHKF   & 0.0100-0.0267 & 0.0100-0.0235 & 0.0100-0.1779 & 0.0300-0.2280\\
\bottomrule
\end{tabular*}}}}
\end{center}

The distributed estimation results for the estimated states are shown in Figs. E.1 (a)-(c), which are obtained from Eqs.(\ref{eq:17}) and (\ref{eq:18}18), as well as the filter parameters given in Theorems 3. Meanwhile, for the established system model using logic ZOH, the linear compensation method for the packet dropouts, and the weighted fusion criteria obtained from the local estimation are investigated. Based on the preceding discussion, the proposed RFHDFE approach possesses the advantage of the better system performance for target tracking and computational efficiency than the state with packet disorders and one-step prediction estimation schemes.

\begin{figure}[H]
\setcounter{subfigure}{0}
{\centering
\subfigure[State estimation for position]{\includegraphics[width=0.8\linewidth]{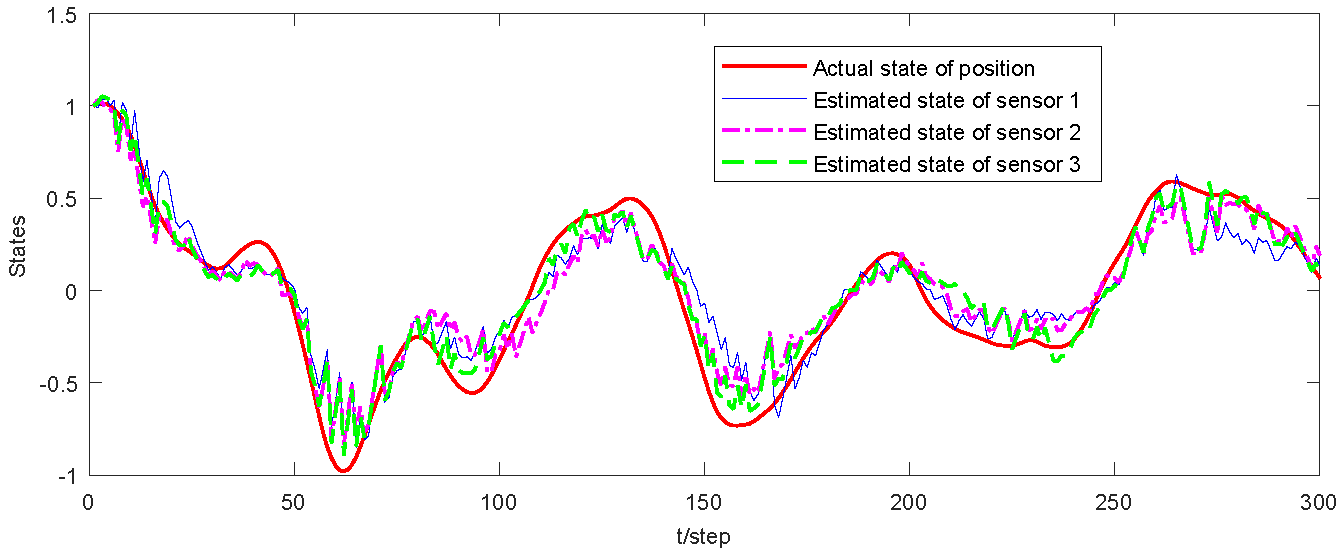}}
\\
\subfigure[State estimation for velocity]{\includegraphics[width=0.8\linewidth]{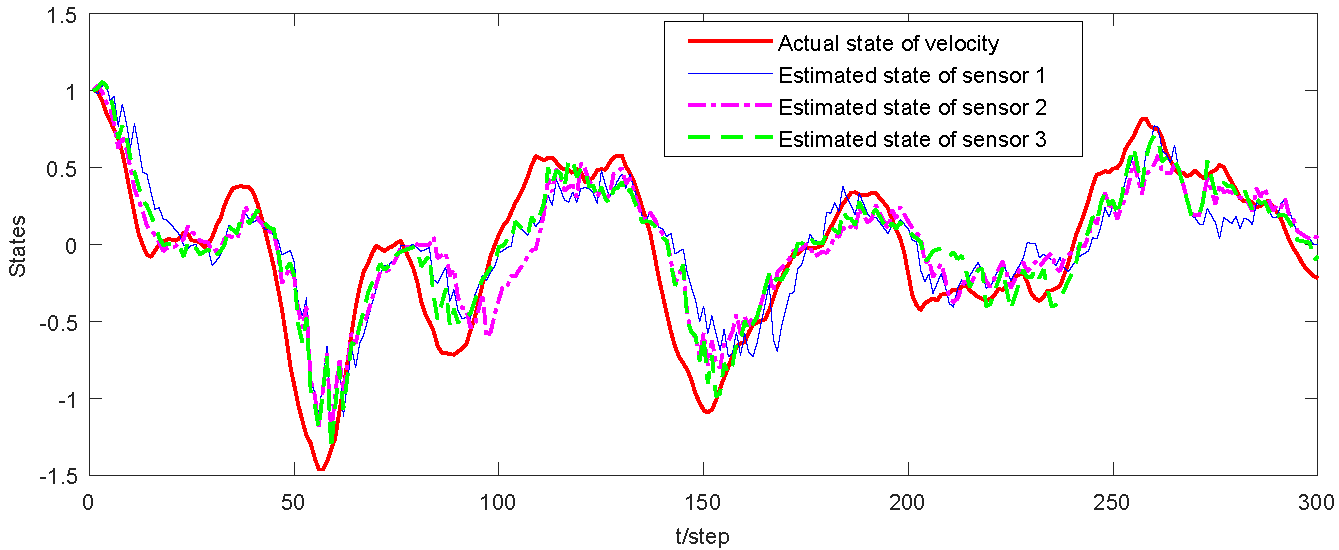}}
\\
\subfigure[State estimation for acceleration]{\includegraphics[width=0.8\linewidth]{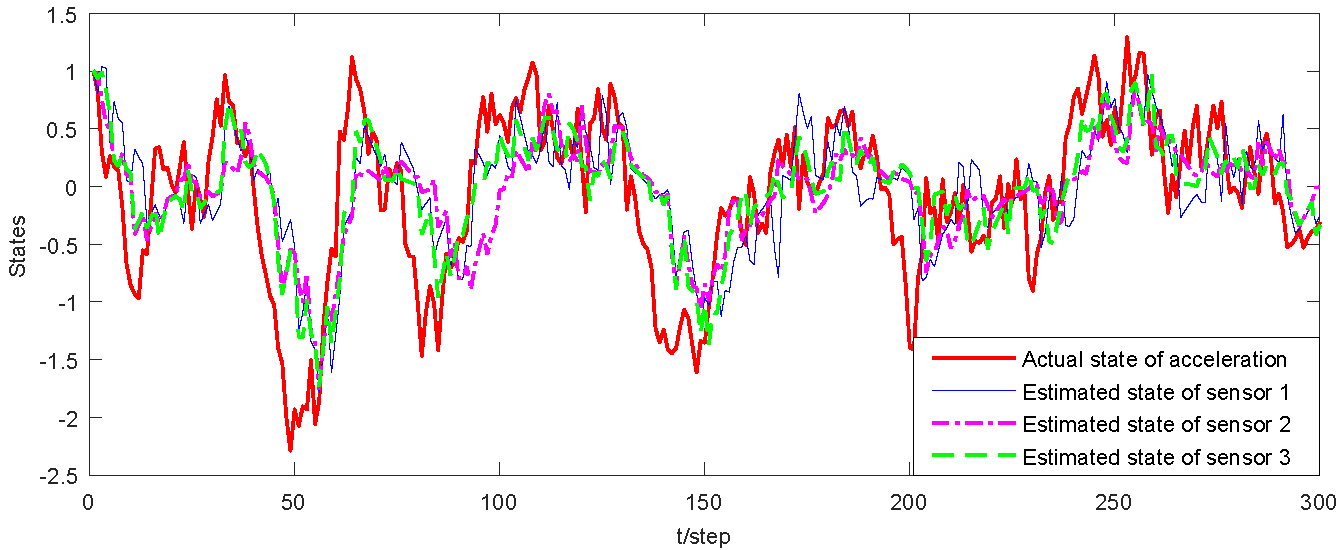}}
\\
\vspace{-4mm}
\caption{Distributed estimation results using RFHDFE}}
\label{Fig E.1}
\end{figure}
%
\bibliographystyle{ieee_fullname}

\end{document}